\documentclass{article}
\usepackage[utf8]{inputenc}
\usepackage{float}
\usepackage{jheppub}
\usepackage{amsmath}
\usepackage[english]{babel}
\usepackage{amssymb}
\usepackage{mathtools}
\usepackage{color}
\definecolor{ao}{rgb}{0.0, 0.5, 0.0}

\usepackage{graphicx}
\usepackage{longtable}
\usepackage{tikz}
\usepackage{comment}
\usepackage{capt-of}
\usepackage{amsthm}
\usepackage{hyperref}
\hypersetup{
    colorlinks=true,
    linkcolor=blue
    }

\usepackage{bm}
\usepackage{tikz}
\usetikzlibrary{backgrounds}
\usepackage{caption}
\usepackage{subcaption}

\usepackage{listofitems} 
\usetikzlibrary{arrows.meta} 
\usepackage[outline]{contour} 
\contourlength{1.4pt}

\tikzset{>=latex} 
\usepackage{xcolor}
\colorlet{myred}{red!80!black}
\colorlet{myblue}{blue!80!black}
\colorlet{mygreen}{green!60!black}
\colorlet{myyellow}{yellow!60!black}
\colorlet{myorange}{orange!85!red!90!black}
\colorlet{mydarkred}{red!30!black}
\colorlet{mydarkblue}{blue!40!black}
\colorlet{mydarkgreen}{green!30!black}
\tikzstyle{node}=[thick,circle,draw=myblue,minimum size=22,inner sep=0.5,outer sep=0.6]
\tikzstyle{node1}=[thick,rectangle,draw=myblue,minimum width = 1cm,  minimum height = 2cm,inner sep=-0.3,outer sep=0.3]
\tikzstyle{node in}=[node,green!20!black,draw=mygreen!30!black,fill=mygreen!8]
\tikzstyle{node hidden}=[node,blue!20!black,draw=myblue!30!black,fill=myblue!10]
\tikzstyle{node convol}=[node,orange!20!black,draw=myorange!30!black,fill=myorange!20]
\tikzstyle{node out}=[node,red!20!black,draw=myred!30!black,fill=myred!20]
\tikzstyle{connect}=[thick,mydarkblue] 
\tikzstyle{connect arrow}=[-{Latex[length=4,width=3.5]},thick,mydarkblue,shorten <=0.5,shorten >=1]
\tikzset{ 
  node 1/.style={node in},
  node 2/.style={node hidden},
  node 3/.style={node out},
}

\newtheorem{remark}{Remark}

\newtheorem{proposition}{Proposition}
\newtheorem{corollary}{Corollary}
\newtheorem{definition}{Definition}
\newtheorem{lemma}{Lemma}
\newtheorem{assumption}{Assumption}

\def\0{\mbox{\tiny $0$}}
\def\1{\mbox{\tiny $1$}}
\def\2{\mbox{\tiny $2$}}
\def\3{\mbox{\tiny $3$}}
\def\4{\mbox{\tiny $4$}}
\def\5{\mbox{\tiny $5$}}
\def\6{\mbox{\tiny $6$}}
\def\7{\mbox{\tiny $7$}}
\def\8{\mbox{\tiny $8$}}
\def\9{\mbox{\tiny $9$}}

\def\etaM{\hat{\eta}}

\def\k{k_{_{B}}}

\def\r{\rangle}
\def\R{\mathcal{R}}
\def\l{\langle}
\def\m{\bar{m}}

\def\q{\bar{q}}
\def\n{\bar{n}}

\def\b{\beta^{'}}

\def\bbt{\tilde{\beta}}

\newcommand{\NotAlpha}{\varphi}

\newcommand{\SOMMA}[2]{\displaystyle\sum\limits_{#1}^{#2}}

\newcommand{\resub}[1]{\textcolor{black}{#1}}

\usetikzlibrary{fit}

\long\def \beq#1\eeq {\begin{equation} #1 \end{equation}}
\long\def \beaq#1\eeaq {\begin{equation}\begin{aligned} #1 \end{aligned}\end{equation}}
\long\def \bes#1\ees {\begin{equation}\begin{split} #1 \end{split} \end{equation}}
\long\def \bea#1\eea {\begin{eqnarray} #1 \end{eqnarray}}
\long\def \bse[#1]#2\ese {\begin{subequations}\label{#1}\begin{align} #2 \end{align}\end{subequations}}

\newcommand{\si}{\sigma_i}

\title{Dense Hebbian neural networks: \\ a replica symmetric picture of unsupervised learning}

\author[a]{Elena Agliari\footnote{Given her role as Editor of this journal, EA had no involvement in the peer-review of articles for which she was an author and had no access to information regarding their peer-review. Full responsibility for the peer-review process for this article was delegated to another Editor.},}
\author[b,f,g]{Linda Albanese,}
\author[b,f]{Francesco Alemanno,}
\author[c]{Andrea Alessandrelli,}
\author[b,f]{Adriano Barra,}
\author[d,e]{Fosca Giannotti,}
\author[c,f]{Daniele Lotito,}
\author[c]{Dino Pedreschi.}

\affiliation[a]{Dipartimento di Matematica, Sapienza Universit\`a di Roma, Piazzale Aldo Moro, 5, 00185, Roma, Italy}

\affiliation[b]{Dipartimento di Matematica e Fisica,  Universit\`a  del Salento, Via per Arnesano, 73100, Lecce, Italy}

\affiliation[c]{Dipartimento di Informatica, Università di Pisa, Lungarno Antonio Pacinotti, 43, 56126, Pisa, Italy}

\affiliation[d]{Scuola Normale Superiore, Piazza dei Cavalieri 7, 56126, Pisa, Italy}

\affiliation[e]{Istituto di Scienza e Tecnologie dell' Informazione, Via Giuseppe Moruzzi, 1, 56124 Pisa, Italy}

\affiliation[f]{Istituto Nazionale di Fisica Nucleare, Campus Ecotekne, Via Monteroni, 73100, Lecce, Italy}

\affiliation[g]{Scuola Superiore ISUFI, Campus Ecotekne, Via Monteroni, 73100, Lecce, Italy}

\abstract{We consider dense, associative neural-networks trained with no supervision and we investigate their computational capabilities analytically, via statistical-mechanics tools, and numerically, via Monte Carlo simulations. In particular, we obtain a phase diagram summarizing their performance as a function of the control parameters (e.g. quality and quantity of the training dataset, network storage, noise) that is valid in the limit of large network size and structureless datasets. Moreover, we establish a bridge between macroscopic observables standardly used in  statistical mechanics and loss functions typically used in the machine learning.
\newline
As technical remarks, from the analytical side, we extend Guerra's interpolation to tackle the non-Gaussian distributions involved in the post-synaptic potentials  while, from the computational counterpart, we insert Plefka's approximation in the Monte Carlo scheme, to speed up the evaluation of the synaptic tensor, overall obtaining a novel and broad approach to investigate unsupervised learning in neural networks, beyond the shallow limit.}

\begin{document}

\maketitle

\section{Introduction} 

Since the seminal work on ``Hebbian Learning'' by John Hopfield in 1982 \cite{Hopfield} and the paradigmatic discoveries  on the landscape of spin-glasses by Giorgio Parisi around the 1980 \cite{MPV}, statistical mechanics has run as a theoretical tool to explain the emergent properties of large assemblies of neurons. The subsequent investigations by Amit, Gutfreund and Sompolinsky \cite{AGS} have definitively established statistical-mechanics as a leading discipline for theoretical studying the collective properties of systems of neurons. In fact, in the last decades, the scientific literature has hosted a large number of contributions on  {\em statistical mechanics of neural networks} (see e.g. \cite{LenkaJPA,LenkaNature,LenkaCarleo}), including countless variations on the original Hopfield model (see e.g. \cite{Auro1,Auro2,Anto1,Anto2,AAAF-JPA2021,Fachechi1,Pierlu1,Pierlu2,Longo}). Among these, Hebbian networks with higher-order interactions, also referred to as {\em dense neural networks} or $P$-spin Hopfield models, were promptly introduced already in the 80's by theoretical physicists (see e.g. \cite{Gardner,Baldi}) as well as by computer scientists (see e.g. \cite{Senio1}).

In the last few years, a renewed interest has raised for dense networks, in fact, on the one hand they have been shown to display intriguing properties as for applications (e.g., they turned out to be robust against adversarial attacks \cite{Krotov2018}, they can perform pattern recognition at prohibitive signal-to-noise level \cite{BarraPRLdetective}) and, on the other hand, many technical issues concerning their analytical and numerical investigation still deserve attention \cite{AuffingerCMP2013, SubagAP2017, SubagProb2017}. Further, recent advances in the usage of pair-wise Hebbian-like networks for learning tasks \cite{EmergencySN,prlmiriam} could be fruitfully extended to the higher-order case. In this work we aim to address these points. 
\newline
More precisely, as for the analytical investigation, we apply interpolating techniques (see e.g., \cite{guerra_broken,Fachechi1}) and extend their range of applicability to include the challenging case of non-Gaussian local fields as it happens for dense networks.
As for the numerical investigation, we propose a strategy based on Plefka expansion \cite{Plefka1,Plefka2} to overcome the strong difficulties implied by the update of a giant synaptic tensor in Monte Carlo (MC) simulations with a remarkable speed up.

Concerning the learning tasks, it should be recalled that, despite its name,  
``Hebbian learning'' (in its standard meaning in Statistical Mechanics, that is provided by the Amit-Gutfreund-Sompolinsky theory of the Hopfield model \cite{Amit}) is actually a \emph{storing} rule and there is no training dataset or inference activity underlying. Yet, recently, the Hebbian rule has been shown to be recastable into a genuine learning rule allowing for both supervised and unsupervised modes and the resulting pair-wise neural network is feasible for a full statistical-mechanics investigation \cite{EmergencySN,prlmiriam}. In this work, we extend this framework to dense neural networks focusing on the unsupervised setting and referring to \cite{super} for the supervised one; the analytical investigations are led under the replica-symmetry (RS) approximation and for structureless datasets.
There are several results stemming from this study.
\newline
First, from an analytical perspective, we are able to summarize the behavior of the network into a phase diagram, namely to highlight in the space of the control parameters of the network (e.g., noise, storage, training-set size, training-set quality) the existence of different regions corresponding to  different computational skills shown by the network. As we will deepen, this is a major reward of the statistical mechanics approach and yields pivotal information towards a {\em sustainable} artificial intelligence since its knowledge allows {\em a priori} setting the machine parameters in the best configuration for a given task. For instance, in this context we can assess the minimal size of the dataset, as function of the dataset quality and the amount of patterns to  store, necessary for a successful training of the machine and we can also estimate the largest amount of information that the machine can safely handle. 
\newline
Second, from a computational perspective, we inserted effective Plefka dynamics in Monte Carlo simulation to obtain a significant speed up for the update of the synaptic tensor (whose handling is otherwise cumbersome). This allowed us to test analytical prediction from the random theory confirming that these networks   -if suitably trained- enjoy pattern recognition capabilities remarkably robust w.r.t. vast amount of noise (as compared to their shallow counterpart) as well as a supra-linear storage of patterns. However, nothing comes for free and the price to pay to enjoy these enhanced information processing capabilities lies in the huge amount of examples that the network has to experience before a correct learning can take place. Thresholds for learning and  maximal storage capabilities also have all been confirmed numerically. 
\newline
Finally, more of conceptual interest, we link quantifiers for machine retrieval (e.g., cost function and Mattis magnetization) to quantifiers for machine learning (e.g., loss function), making these two capabilities of neural networks -learning and retrieval- two aspects of a unified cognitive process and yielding a cross-fertilization between the two related fields (i.e., statistical mechanics and machine learning).

The paper is structured as follows:  in Sec.~\ref{definitions} we briefly review the state-of-the-art on Hebbian learning and in Sec.~\ref{sec:dense_unsup} we introduce the unsupervised dense Hebbian networks and define the related macroscopic observables necessary for its investigation. Next, in Sec.~\ref{sec:loss}, we discuss the connection between performance quantifiers stemming from, respectively, statistical-mechanics and machine learning. Then, in Sec.~\ref{sec:solution}, we study the model exploiting the statistical-mechanics framework, while the numerical investigation is addressed in Sec.~\ref{sec:num}, by relying on Monte Carlo simulations. Finally, in Sec.~\ref{sec:conclusions} we summarize and discuss our results. Technical details are collected in the Appendices.

Finally we stress again that the present manuscript is dedicated to the study of unsupervised learning, while the supervised protocol for dense networks is addressed in a twin work \cite{super}.


\section{Prelude: from Hebbian storing to Hebbian learning}\label{definitions}
The standard Hopfield model is built upon $N$ binary neurons, denoted as $\boldsymbol \sigma = (\sigma_1, \sigma_2, ..., \sigma_N)  \in \{ -1, +1 \}^N$, that are employed to reconstruct the information encoded in $K$ binary vectors of length $N$, also called \emph{patterns} and denoted as $\boldsymbol \xi^{\mu} \in \{-1, +1 \}^N$ with $\mu= 1,\hdots,K$, whose entries are Rademacher random variables, namely, for the generic $(i,\mu)$ entry
\begin{equation} \label{eq:rademacher}
    \mathbb{P}(\xi_{i}^{\mu}) =  \frac{1}{2} \left [ \delta_{\xi_{i}^{\mu}, -1} + \delta_{\xi_{i}^{\mu}, +1} \right].
\end{equation}
This information is allocated in the synaptic matrix, namely in the couplings $\boldsymbol J = \{J_{ij}\}_{i,j=1,...,N}$ among neurons, defined according to Hebb's rule
\begin{equation}\label{eq:hebb}
J_{ij}=\frac{1}{N}\sum_{\mu=1}^K \xi_i^{\mu}\xi_j^{\mu},
\end{equation}
ensuring that, under suitable conditions, the system can play as an associative memory (\emph{vide infra}).
The network has a Hamiltonian cost-function representation 
$$
\mathcal H_{N,K}^{\textrm{(H)}}(\boldsymbol \sigma| \boldsymbol \xi)= -\sum_{i<j}^{N,N}J_{ij}\sigma_i\sigma_j = -\frac{N}{2} \sum_{\mu=1}^K m_{\mu}^2+\dfrac{K}{2},
$$
where in the last equivalence we used \eqref{eq:hebb} and we introduced the Mattis magnetization 
\begin{equation} \label{eq:Mattis}
m_{\mu}:=\frac{1}{N} \sum_{i=1}^N \xi_i^{\mu}\sigma_i.
\end{equation}
The occurrence of a neural configurations $\boldsymbol \sigma$ is ruled by the Boltzmann-Gibbs probability $\propto e^{-\beta \mathcal H_{N,K}^{\textrm{(H)}}(\boldsymbol \sigma | \boldsymbol \xi) }$
where $\beta:=1/T$ tunes the degree of stochasticity and, in a physical jargon, represents the inverse of the temperature.
\newline
The relaxation to a state where $m_{\mu}$ is close to $1$ is interpreted as the \emph{retrieval} of the pattern $\boldsymbol \xi^{\mu}$. 

As anticipated, the statistical-mechanical analysis allows summarizing the performance of a network into a phase diagram, which highlights the existence of qualitatively different behaviors of the system as its control parameters are tuned. Here the control parameters are the above-mentioned temperature $T$, also referred to as ``fast noise'', and the load $\alpha$ defined as $\alpha:= \lim_{N \to \infty} K/N$, also referred to as ``slow noise''. The phase diagram for the Hopfield model (see e.g., \cite{AGS,Coolen}) is reported in Fig.~\ref{fig:diagramma} (left panel): one can notice that this machine is able to work as an associative memory solely in the retrieval region, corresponding to loads $\alpha<\alpha_c \approx 0.14$. Thus, it is pointless trying to allocate a larger amount of patterns as the machine will not be able to retrieve them: the knowledge of the phase diagram allows us to use this information \emph{a priori}, before any trial is performed, thus potentially saving energy and CPU time\footnote{Optimized protocols in AI are especially longed for as, at present, training AI on large scale can result in conflicts with green policies \cite{MITpress}.}.

 \begin{figure}[t]
 \centering
    \includegraphics[scale=0.8]{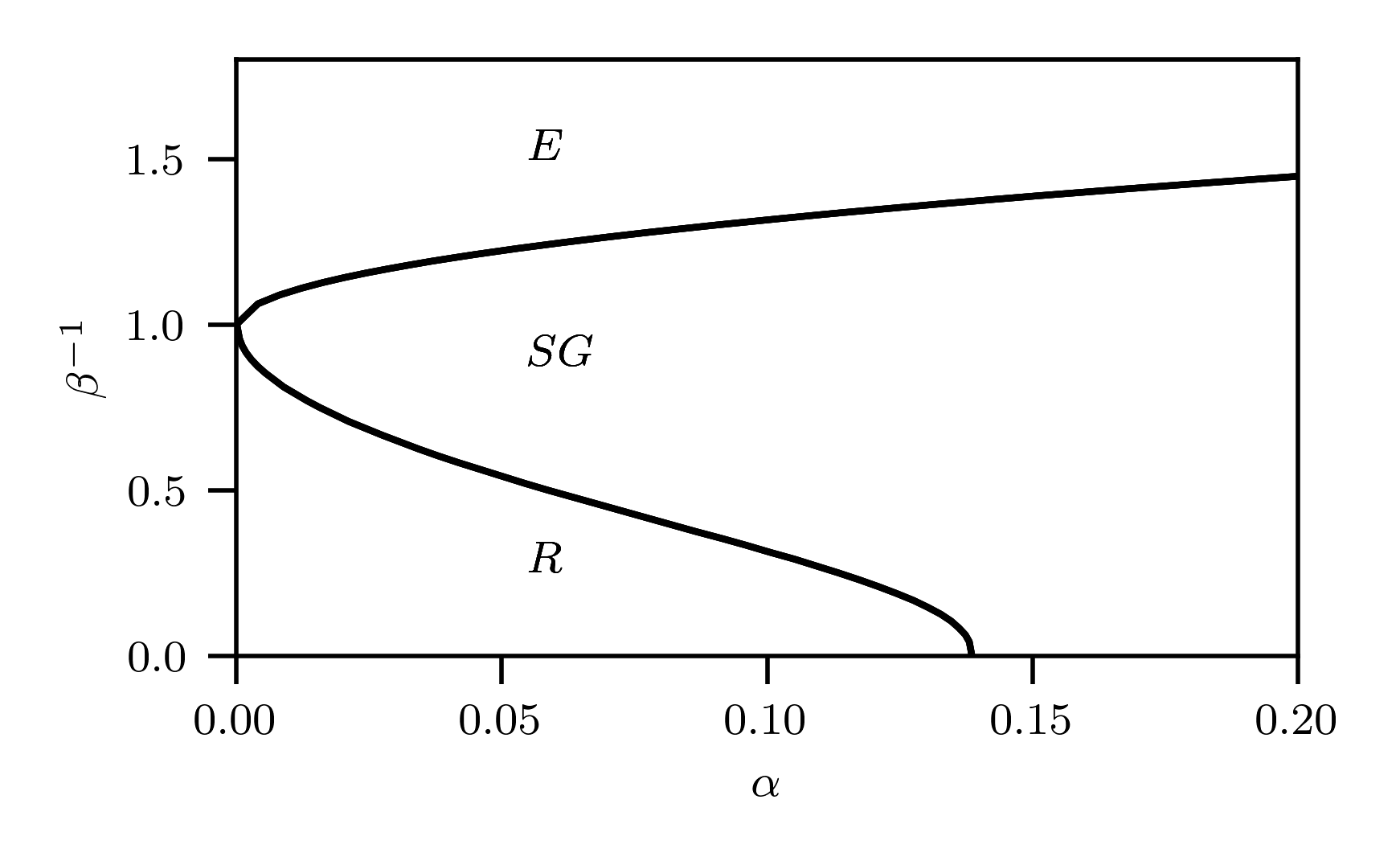}
\includegraphics[scale=0.8]{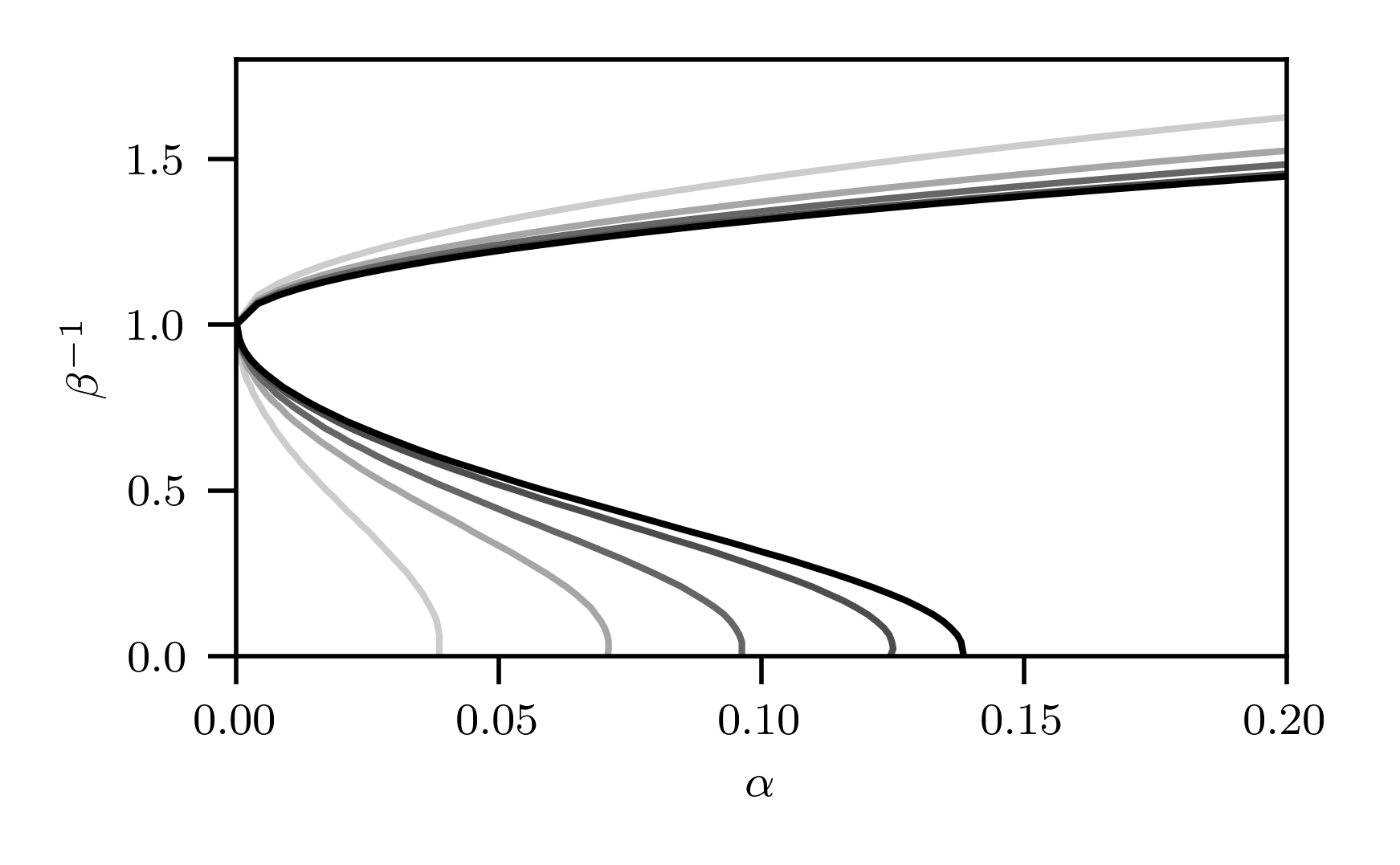}
    \caption{Left: Phase diagram of the Hopfield model. Three regions corresponding to qualitatively different behaviors of the system are highlighted: ergodic (E, where fast noise prevails), spin-glass (SG, where slow noise prevails) and retrieval [R, where (a close neighbourhood of) each pattern $
\boldsymbol \xi^{\mu}$ plays as an attractor for the neural configuration and consequently the system can perform pattern recognition as an associative memory]. In particular, in the retrieval region, if a noisy example of a pattern, say $\boldsymbol \eta^{\mu}$, is inputted to the network (namely the neuron configuration is initialized as $\boldsymbol \sigma = \boldsymbol \eta^{\mu}$), the latter reconstructs the original pattern (namely $\boldsymbol \sigma$ spontaneously relaxes (close) to $\boldsymbol \xi^{\mu}$). 
Right: Phase diagram of the unsupervised Hopfield model. The darker the lines, the better the quality $r$ of the supplied dataset ($r = 0.4, 0.5, 0.6, 1$).
As the quality of the dataset improves, the retrieval region expands, and when the dataset quality saturates to $1$ (black line) the phase diagram recovers the one in the left panel. Here $M=40$ and analogous results are found by retaining $r$ constant and increasing $M$.
}
    \label{fig:diagramma}
\end{figure}

Despite the expression \eqref{eq:hebb} is often named Hebbian {\em learning}, the above model has little to share with machine learning as there are no real learning processes underlying. These could however be introduced with minimal modification of Eq.~\eqref{eq:hebb}, as we are going to explain. Let us treat each pattern $\boldsymbol \xi^{\mu}$ as an \emph{archetype} and use it to generate a training set of $M$ \emph{examples} for each archetype. Denoting with $\boldsymbol \eta^{\mu,a} \in \{-1, +1\}^N$ the $a$-{th} example of the $\mu$-{th} archetype, we can write two generalizations of the above Hebbian rule, namely
\begin{eqnarray}
\label{eq:J_unsup}
J_{ij}^{(unsup)} &=& \frac{1}{NM}\sum_{\mu=1}^K \sum_{a=1}^M \eta_i^{\mu,a}\eta_j^{\mu,a},\\  
\label{eq:J_sup}
    J_{ij}^{(sup)} &=&\frac{1}{NM^2}\sum_{\mu=1}^K \Big(\sum_{a=1}^M \eta_i^{\mu,a} \Big) \Big(\sum_{a=1}^M \eta_j^{\mu,a}\Big),
\end{eqnarray}
where, in the first expression there is no teacher that knows the labels and can cluster the examples archetype-wise as it happens in the second scenario, this is why the two generalizations are associated to, respectively, unsupervised and supervised protocols \cite{EmergencySN,prlmiriam}.
\newline
In order to build our dataset $\{ \boldsymbol \eta^{\mu,a} \}_{a=1,...,M}^{\mu=1,...,K}$, we generate $M$ randomly-perturbed copies of each archetype, interpreted as {\em examples} and whose entries $(i,\mu,a)$ are described by
\begin{equation} \label{eq:Bernoulli}
    \mathbb{P}(\eta_{ i}^{\mu,a}|\xi^\mu_i) = \frac{1-r}{2} \delta_{\eta_{ i}^{\mu,a},-\xi_{ i}^{\mu}} + \frac{1+r}{2} \delta_{\eta_{ i}^{\mu,a},\xi_{ i}^{\mu}}, 
\end{equation}
in such a way that $r \in [0,1]$ assesses the training-set \emph{quality}, that is, as $r \rightarrow 1$ the example matches perfectly the archetype, whereas for $r \to 0$ an example is, in the average, orthogonal to the related archetype. \\
A natural question is thus wondering the existence of a threshold $M_\otimes$ beyond which the network can certainly infer the archetypes that gave rise to a newly experienced example.
A schematic representation to figure out how the learning process of the archetype works in this kind of network is provided in Fig.~\ref{fig:cartoon}. 

The unsupervised, pairwise Hopfield model supplied with this kind of dataset has been investigated in details in \cite{EmergencySN,prlmiriam}, obtaining a full statistical-mechanics description summarized in the phase diagram reported in Fig.~\ref{fig:diagramma} (right panel). Interestingly, one can see that, as the dataset is impaired (because either $r$ or $M$ is reduced), the retrieval region shrinks.

A useful quantity to assess the overall information content of the dataset $\{ \boldsymbol \eta^{\mu,a} \}_{a=1,...,M}^{\mu=1,...,K}$
is given by $\rho=\frac{1-r^2}{M r^2}$, which in the following shall be referred to as \emph{dataset entropy}. Strictly speaking, $\rho$ is not an entropy, yet here we allow ourselves for this slight abuse of language because, as discussed in \cite{prlmiriam}, the conditional entropy, that quantifies the amount of information needed to describe the original message $\boldsymbol{\xi}^\mu$ given the set of related examples $\{\boldsymbol{\eta}^{\mu,a}\}_{a=1,...,M}$, is a monotonically increasing function of $\rho$.

 \begin{figure}[t] 
    \centering
    \includegraphics[width=14cm]{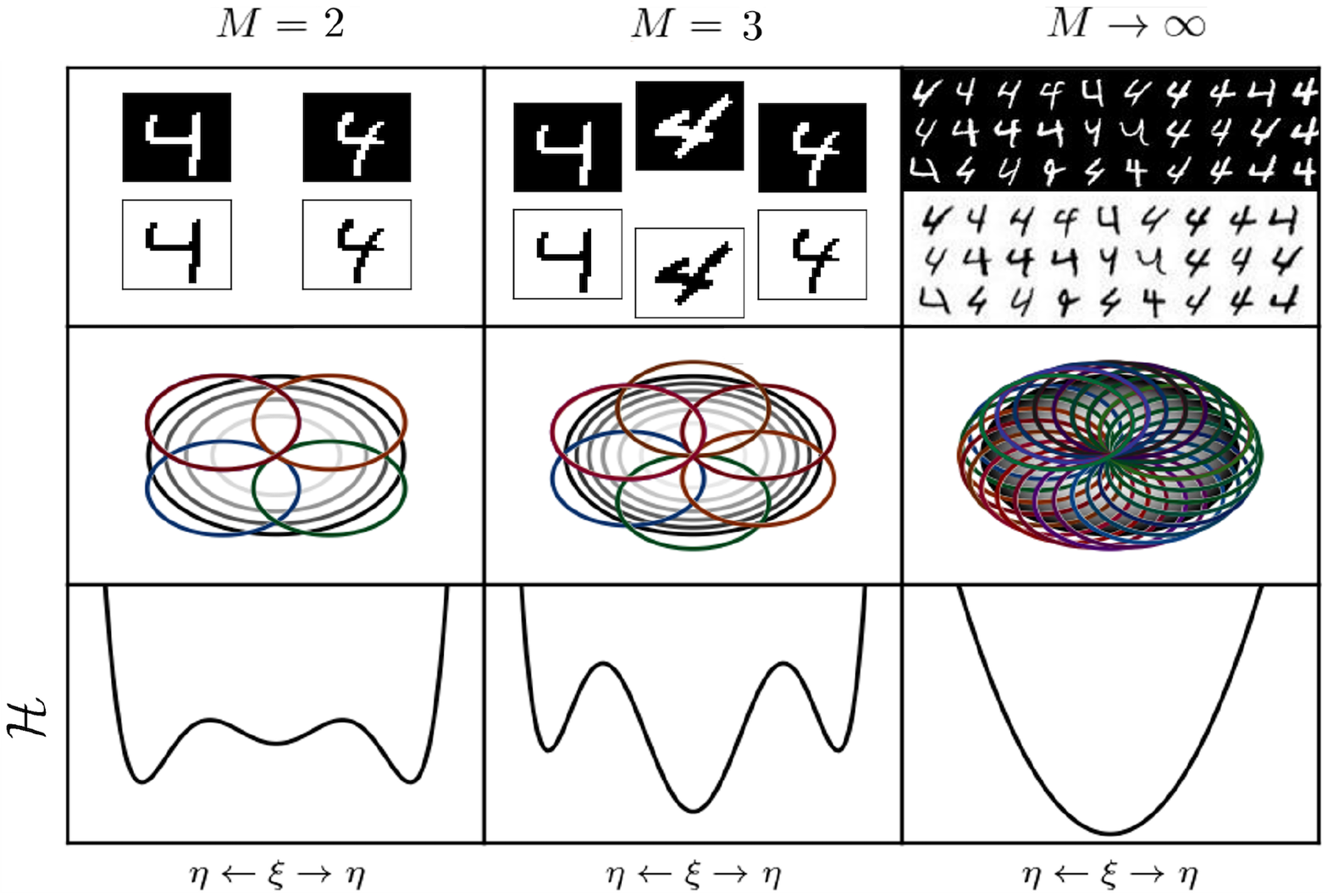}
    \caption{Intuitive representation of the process of learning an archetype. In the upper row we show the neuron configurations corresponding to the supplied examples, in the middle row we schematically depict the attraction basins determined by these examples, and in the lower row we sketch a plausible cross-sectional view of the Hamiltonian function in a fictitious one dimensional representation of the configuration space.
    Going from left to right, as the number of examples $M$ grows, the network learns to generalize from them by constructing a faithful representation of the generic archetype $\boldsymbol \xi$ (that has never been supplied to the network). The network tends to store at first each single example $\boldsymbol \eta^a$ without being able to retrieve the archetype (left column), thus the deepest minima of the Hamiltonian correspond to examples. Then, a minimum, close to $\boldsymbol \xi$, appears and coexists with the other minima (middle column) and, finally, a unique stable minimum corresponding to the archetype emerges (right column). 
    The variation of the energy landscape as $M$ changes depends on the network architecture and on the dataset. 
    }
    \label{fig:cartoon}
\end{figure}

\section{Dense Hebbian neural networks in the unsupervised setting} \label{sec:dense_unsup}

We consider a network with $N$ Ising  neurons $\si \in  \{ -1, + 1 \}$ with $i= {1,\hdots ,N}$, $K$ Rademacher archetypes $\bm \xi^{\mu} \in \{ -1, + 1 \}^N$ and $M$ noisy examples $\bm \eta^{\mu,a} \in \{-1, +1\}^N$ per archetype  $\bm \xi^{\mu}$ with $a= 1,\hdots ,M$ and $\mu =1,\hdots, K$, whose entries are drawn according with, respectively, \eqref{eq:rademacher} and \eqref{eq:Bernoulli}.
\newline
In the network considered here interactions among neurons are $P$-wise and their magnitude is obtained by generalizing \eqref{eq:J_unsup}, as captured by the next
\begin{definition} \label{def:dense_unsup}
The cost-function (or {\em Hamiltonian}) of the dense Hebbian neural network in the unsupervised regime is 
\begin{align}
    \mathcal{H}_{N,K,M,r}^{(P)}(\bm \sigma \vert \bm \eta) =& -\dfrac{1}{\R^{P/2}\,M\,N^{P-1}}\SOMMA{\mu=1}{K}\SOMMA{a=1}{M}\left(\SOMMA{(i_1,\hdots,i_P)}{N,\hdots,N}\eta^{\mu\,,a}_{i_1}\cdots\eta^{\mu\,,a}_{i_P}\sigma_{i_1}\cdots\sigma_{i_P}\right),
    \label{def:H_PHopEx}
\end{align}
where $P$ is the interaction order (assumed as even), $\R=r^2+\frac{1-r^2}{M}$ corresponds to $\mathbb{E}_{\xi}\mathbb{E}_{(\eta \vert \xi)}\left[\sum_a \eta_i^{\mu,a}/(Mr)\right]^2$ (and plays as a normalization factor) and  we also define  $\sum\limits_{(i_1,\hdots,i_P)}^{N,\hdots,N} =\sum\limits_{\underset{i_1\neq\hdots\neq i_P}{i_1,\hdots,i_P}}^{N,\hdots,N}$ (namely the summation in which only terms with all different "$i$" indices are taken into account). Further, the factor $\frac{1}{N^{P-1}}$ in the right-hand side ensures the linear extensiveness of the Hamiltonian in the network size $N$. 
\newline
The related partition function is defined as 
\begin{align}
    \mathcal{Z}^{(P)}_{N,K,M,r,\beta }(\bm \eta) = \sum_{\bm \sigma}^{2^N} \exp \left( -\beta \mathcal{H}^{(P)}_{N,K,M,r}(\bm \sigma \vert \bm\eta)\right) =: \sum_{\bm \sigma}^{2^N} \mathcal{B}^{(P)}_{N,K,M,r, \beta}(\bm \sigma  \vert  \bm\eta),
    \label{eq:Def_orig_Z}
\end{align}
    where $\mathcal{B}^{(P)}_{N,K,M,r, \beta}(\bm \sigma \vert  \bm\eta)$ is referred to as Boltzmann factor.
\newline 
At finite network-size $N$, the quenched statistical  pressure (or free energy\footnote{The free energy $\mathcal F^{(P)}_{N,K,M,r,\beta}$ equals the statistical pressure, a factor $-\beta$ apart, i.e. $\mathcal A^{(P)}_{N,K,M,r,\beta }= - \beta  \mathcal F^{(P)}_{N,K,M,r,\beta}$. Thus, extremizing the former results in the same self-consistency equations for the macroscopic observables that we would obtain by extremizing the latter; in this paper we use mainly the statistical pressure with no loss of generality.} with no loss of generality.) of the model reads as 
    \begin{align}
        \mathcal A^{(P)}_{N,K,M,r,\beta} = \frac{1}{N}\mathbb{E}\log \mathcal{Z}^{(P)}_{N,K,M,r, \beta}(\bm\eta)
        \label{PressureDef_unsup}
    \end{align}
    where $\mathbb{E}=\mathbb{E}_{\xi}\mathbb{E}_{(\eta|\xi)}$ denotes the average over the realization of examples, namely over the distributions \eqref{eq:rademacher} and \eqref{eq:Bernoulli}.
    By combining the quenched average $\mathbb{E}[\cdot]$ and the Boltzmann average
    \begin{equation} \label{omegaNKM}
    \omega[(\cdot)]:= \frac{1}{\mathcal{Z}^{(P)}_{N,K,M,r, \beta}( \bm \eta)} \sum_{\boldsymbol \sigma}^{2^N} ~ (\cdot) ~ \mathcal{B}^{(P)}_{N,K,M,r, \beta}(\bm \sigma \vert  \bm\eta),
    \end{equation}
    possibly replicated over two or more replicas\footnote{A replica is an independent copy of the system characterized by the same realization of disorder, namely by the same realization of the archetypes and examples. Thus, two replicas are sampled from the same distribution $\mathcal B^{(P)}_{N,K,M,r,\beta}(\boldsymbol \sigma| \bm\eta)$.
    Comparing two copies allows us to determine whether slow noise prevails, that is, whether the interference between archetypes and examples prevents the system to retrieve, see Sec.~\ref{sec:solution}}, that is, $\Omega:=\omega \times \omega ...\times \omega$, we get the expectation
       \begin{equation}
    \langle \cdot \rangle := \mathbb{E} \Omega(\cdot).
     \end{equation}
\end{definition}

\begin{remark}
An integral representation of the partition function will be useful in the following numerical computations.
Starting from Eq. \eqref{eq:Def_orig_Z}, we apply the Hubbard-Stratonovich transformation to get

\small
\begin{align} \label{eq:integral}
    &\mathcal{Z}^{(P)}_{N,K,M ,r,\beta} (\bm\eta) =  \sum_{\bm \sigma} \int \prod_{\mu,a} d \tilde{\mu}(z_{\mu,a}) \exp \left[ \sqrt{\dfrac{\beta '}{\R^{P/2} M\,N^{P-1}}} \, \SOMMA{\mu>1}{K}\SOMMA{a=1}{M} \SOMMA{ \ i_{_1},\cdots ,i_{_{P/2}}}{N,\cdots,N}\eta^{\mu\,,a}_{i_1}\cdots\eta^{\mu\,,a}_{i_{P/2}}\:\sigma_{_{i_1}}\cdots\sigma_{_{i_{_{P/2}}}}\,z_{_{\mu\,,a}}\right]
\end{align}
\normalsize
where $d \tilde{\mu} (z_{\mu,a}) = \dfrac{\exp(-z_{\mu,a}^2/2)}{\sqrt{2\pi}} dz_{\mu,a}$ is a Gaussian measure and we posed $\beta'= 2\frac{\beta}{P!}$. Moreover, we have exploited
\begin{equation}
    \begin{array}{lll}
         &\dfrac{P!}{2 N^{P-1}}\SOMMA{(i_1,\cdots i_{P})}{N,\cdots,N}\,\left( \Phi_{i_1}^{\mu} \hdots \Phi_{i_{P}}^{\mu} \right)=\dfrac{1}{2 N^{P-1}}\SOMMA{i_1,\cdots i_{P}}{N,\cdots,N}\,\left( \Phi_{i_1}^{\mu} \hdots \Phi_{i_{P}}^{\mu} \right)+ \mathcal{O}(N^{P/2-1})\,
    \end{array}
\end{equation}
with $\Phi_i^\mu$ is any finite random variable and we have neglected the subleading network-size terms.
\newline
We can think of the above transformation as a mapping between  the original dense Hebbian network and a restricted Boltzmann machine (RBM)  where $K\times M$  hidden neurons $z_{\mu,a}$ (equipped with a Gaussian prior) interact with the $N$ visible neurons $\boldsymbol{\sigma}$ grouped in sets each made of $P/2$ neurons $\sigma_{i_1}\cdots\sigma_{i_{P/2}}$ with weight $\eta_{i_1}^{\mu,a_1} \cdots \eta_{i_{P/2}}^{\mu,a_{P/2}}$. A schematic representation of the dense Hebbian network  and its dual RBM are shown for a simple case in Fig. \ref{fig:network}.
\end{remark}

\begin{figure}[t]
\centering
\begin{tikzpicture}[x=2.0cm,y=1.2cm]

\node[node in,outer sep=0.6] (S-1) at (-3.5,0.2) {$\sigma_{1}$};
\node[node in,outer sep=0.6] (S-2) at (-2.5,0.2) {$\sigma_{2}$};
\node[node in,outer sep=0.6] (S-3) at (-4.0,-1.1) {$\sigma_{3}$};
\node[node in,outer sep=0.6] (S-4) at (-2.0,-1.1) {$\sigma_{4}$};
\node[node in,outer sep=0.6] (S-5) at (-3.5,-2.4) {$\sigma_{5}$};
\node[node in,outer sep=0.6] (S-6) at (-2.5,-2.4) {$\sigma_{6}$};

\draw[line width=1.5pt] (S-1) -- (S-5);
\draw[line width=1.5pt] (S-1) -- (S-2);
\draw[line width=1.5pt] (S-2) -- (S-3);
\draw[line width=1.5pt] (S-3) -- (S-5);

\begin{scope}[on background layer]
\path [fill=lightgray,opacity=0.9,very thin] (S-1.center) to  (S-5.center) 
    to  (S-3.center) to (S-2.center);
    \begin{scope}[on background layer]
    \draw[line width=1.5pt, red] (S-1) -- (S-6);
    \draw[line width=1.5pt, red] (S-2) -- (S-5);
    \draw[line width=1.5pt, red] (S-5) -- (S-6);
        \path [fill=myred!20,opacity=0.6,very thin] (S-1.center) to  (S-2.center) 
    to  (S-5.center) to (S-6.center);
    \begin{scope}[on background layer]
    \draw[line width=1.5pt, violet] (S-6) -- (S-2);
    \draw[line width=1.5pt, violet] (S-4) -- (S-1);
    \draw[line width=1.5pt, violet] (S-4) -- (S-6);
        \path [fill=violet!70,opacity=0.3,very thin] (S-1.center) to  (S-2.center)  to  (S-6.center) to (S-4.center);
    \end{scope}
    \end{scope}
\end{scope}

 \draw[<-]        (-3.8,-0.0)node[left, scale = 0.8] {\large$\dfrac{1}{M}\SOMMA{\mu,a=1}{2,3}\eta_1^{\mu,a}\eta_2^{\mu,a}\eta_3^{\mu,a}\eta_5^{\mu,a}$}   -- (-3.5,-0.7) ;

 \draw[->, violet]        (-2.5,-0.7)   -- (-2.3,0.0) node[right, scale = 0.8]{\large$\dfrac{1}{M}\SOMMA{\mu,a=1}{2,3}\eta_1^{\mu,a}\eta_2^{\mu,a}\eta_4^{\mu,a}\eta_6^{\mu,a}$};

 \draw[->, red]        (-3.0,-1.1)   -- (-3.0,-2.9) node[below, scale = 0.8] {\large$\dfrac{1}{M}\SOMMA{\mu,a=1}{2,3}\eta_1^{\mu,a}\eta_2^{\mu,a}\eta_5^{\mu,a}\eta_6^{\mu,a}$};

\node[node in,outer sep=0.6] (NI-1) at (0,0.5) {$\sigma_{1}$};
\node[node in,outer sep=0.6] (NI-3) at (0,-0.5) {$\sigma_{3}$};
\node[node in,outer sep=0.6] (NI-4) at (0,-1.5) {$\sigma_{4}$};
\node[node in,outer sep=0.6] (NI-5) at (0,-2.5) {$\sigma_{6}$};
    
\node[node hidden] (NO-1) at (1,1.35) {$z_{1,1}$};
\node[node hidden] (NO-2) at (1,1.35-0.8) {$z_{1,2}$};
\node[node hidden] (NO-3) at (1,1.35-1.6) {$z_{1,3}$};
\node[node hidden] (NO-4) at (1,-1.8) {$z_{2,1}$};
\node[node hidden] (NO-5) at (1,-1.8-0.8) {$z_{2,2}$};
\node[node hidden] (NO-6) at (1,-2.6-0.8) {$z_{2,3}$};

\node[draw, line width =0.035cm,dotted,fit=(NO-1) (NO-2) (NO-3)] {};

\node[draw, line width =0.035cm,dotted,fit=(NO-4) (NO-5) (NO-6)] {};

\node[above=9, right=15,align=center] at (NO-1) {$\mu=1$};
\node[below=9, right=15,align=center] at (NO-6) {$\mu=2$};
  
  \draw[connect, line width=1.5pt] (NI-3) -- (NO-4);
  \draw[connect, line width=1.5pt] (NI-4) -- (NO-4);
  \draw[connect, line width=1.5pt] (NI-3) -- (NI-4);
  \draw[connect, line width=1.5pt] (NI-3) -- (NO-3);
  \draw[connect, line width=1.5pt] (NI-4) -- (NO-3);
  \draw[connect, line width=1.5pt] (NI-3) -- (NO-5);
  \draw[connect, line width=1.5pt] (NI-4) -- (NO-5);
  \path (NI-4) --++ (NI-5) node[midway,scale=0.8] {$\vdots$};
  \path (NI-1) --++ (NI-3) node[midway,scale=0.8] {$\vdots$};

\begin{scope}[on background layer]
    \path [fill=myorange,opacity=0.6,very thin] (NI-3.center) to  (NI-4.center) 
        to  (NO-5.center);
    \draw[-, myorange]          (1.4,-2.15)node[right, scale = 0.8] {\large$\eta_3^{2,2}\eta_4^{2,2}$}  -- (0.6,-2.2);

    \path [fill=mygreen,opacity=0.6,very thin] (NI-3.center) to  (NI-4.center) 
        to  (NO-3.center);
    \draw[-, mydarkgreen]        (0.5,-0.9)   -- (1.4,-0.5) node[right, scale = 0.8] {\large$\eta_3^{1,3}\eta_4^{1,3}$};

    \path [fill=cyan,opacity=0.6,very thin] (NI-3.center) to  (NI-4.center) 
        to  (NO-4.center);
    \draw[-, cyan!60!black]        (0.5,-1.6)   -- (1.4,-1.1) node[right, scale = 0.8] {\large$\eta_3^{2,1}\eta_4^{2,1}$};
\end{scope}
  \node[right,scale=0.9] at (-1.05,-1.3)
    {$\Longleftrightarrow$};

    \node[above=60, right=17,align=center] at (NI-1) {RBM};
\node[below=18, align=center] at (NI-5) {visible};
\node[below=130,align=center] at (NO-3) {hidden};
\node[above=25, right=15,align=center] at (S-1) {NN};
\end{tikzpicture}
\caption{Representation of a dense unsupervised Hopfield network (NN, left) and its dual Restricted Boltzmann Machine (RBM, right), with $N=6$, $K=2$, $M=3$ and $P=4$. Different groups of interacting neurons are depicted in different colors.
As far as the RBM concerns, it is built with a visible layer made of $N$ binary variables $\{\sigma_i\}_{i=1,\hdots,6}$ and a hidden layer made of $K\times M$ Gaussian neurons $\{z_{\mu,a}\}^{a=1,2,3}_{\mu=1,2}$. In particular, any $z_{\mu,a}$ can interact with sets of $P/2$ (namely, $2$ in this case) visible neurons $\{\sigma_i, \sigma_j\}$ whose strength of interaction is $\eta_i^{\mu,a}\eta_j^{\mu,a}$. In the NN, the neurons interact $4-$wise and the coupling strength for any set of variables $\{\sigma_i, \sigma_j,\sigma_k,\sigma_l\}$ is $\frac{1}{M}\sum\limits_{\mu=1}^{K}\sum\limits_{a=1}^{M}\eta_i^{\mu,a}\eta_j^{\mu,a}\eta_k^{\mu,a}\eta_l^{\mu,a}$.}
\label{fig:network}
\end{figure}
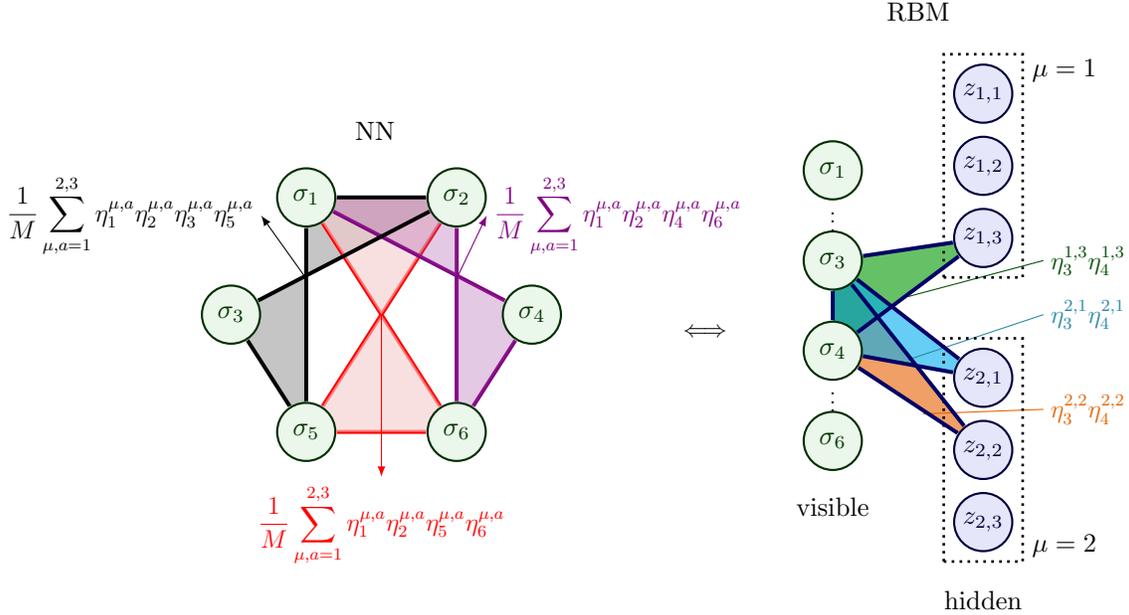

In our analytical investigation we leverage the asymptotic limit for the system size $N$, which shall be performed retaining the network load $\alpha_{b}$ finite, as specified by the following
  \begin{definition}
    In the thermodynamic limit $N \to \infty$, the load is defined as
\begin{equation}
\lim_{N\to +\infty} \frac{K}{N^{b}} =: \alpha_{b} < \infty
\label{eq:carico_TDL}
\end{equation}
with $b \leq P-1$\footnote{The case $b>P-1$ is known to lead to a black-out scenario \cite{Baldi,Bovier} not useful for computational purposes and shall be neglected here.}.  We then distinguish between
the so-called {\em high-load regime}, corresponding to $b = P-1$, namely to an amount of storable patterns that scales with the networks size as  $N^{P-1}$, and a so-called {\em low-load regime} corresponding to $b < P-1$. As we will deepen, the resulting slower scaling for the amount of storable patterns allows for mitigating the effects of possible additive noise affecting synaptic strengths (see Sec. \ref{sec:ultra-noise}). 

Further, the quenched statistical pressure in the thermodynamic limit is denoted as
\begin{equation} \label{eq:statpress_LTD}
\mathcal A^{(P)}_{\alpha_{b}, M, r, \beta}  = \lim_{N \to \infty} \mathcal A^{(P)}_{N,K,M, r,\beta} .
\end{equation}
\end{definition}

In order to further simplify the notation it is convenient to introduce a $P$-independent load denoted as $\gamma$ and defined by
\begin{equation}
    \alpha_{P-1}=\gamma \dfrac{2}{P!}.
    \label{eq:alphaPP-1}
\end{equation}
We can notice that, as long as $P$ is fixed, assuming that $\alpha_{P-1} < \infty$ also means that $\gamma < \infty$. 

\par\medskip
We want to study the model defined in \ref{def:dense_unsup} looking at its learning and retrieval capabilities and, specifically, we aim to find out the thresholds for these capabilities to emerge. 
In other words, given a training dataset made of $M \times K$ examples, each codified by a binary vector of size $N$ and characterized by a quality $r$, and given a set of $N$ binary neurons that interact $P$-wisely, we aim at answering the following questions: 
\begin{itemize}
\item which is the minimum number of examples to be supplied to the network to ensure that it is able to infer the related archetypes and thus correctly generalize afterwards? (Note that we address this question while the network is handling simultaneously all the $K \times M$ archetypes.)

\item how many archetypes the network can learn and what happens if we load the network with a larger amount of information?

\item can we account for training flaws in this system and, if so, how robust is the resulting pattern recognition capability of the network with respect to this kind of noise?
\end{itemize}

\begin{figure}[t]
    \centering
    \includegraphics[scale=0.42]{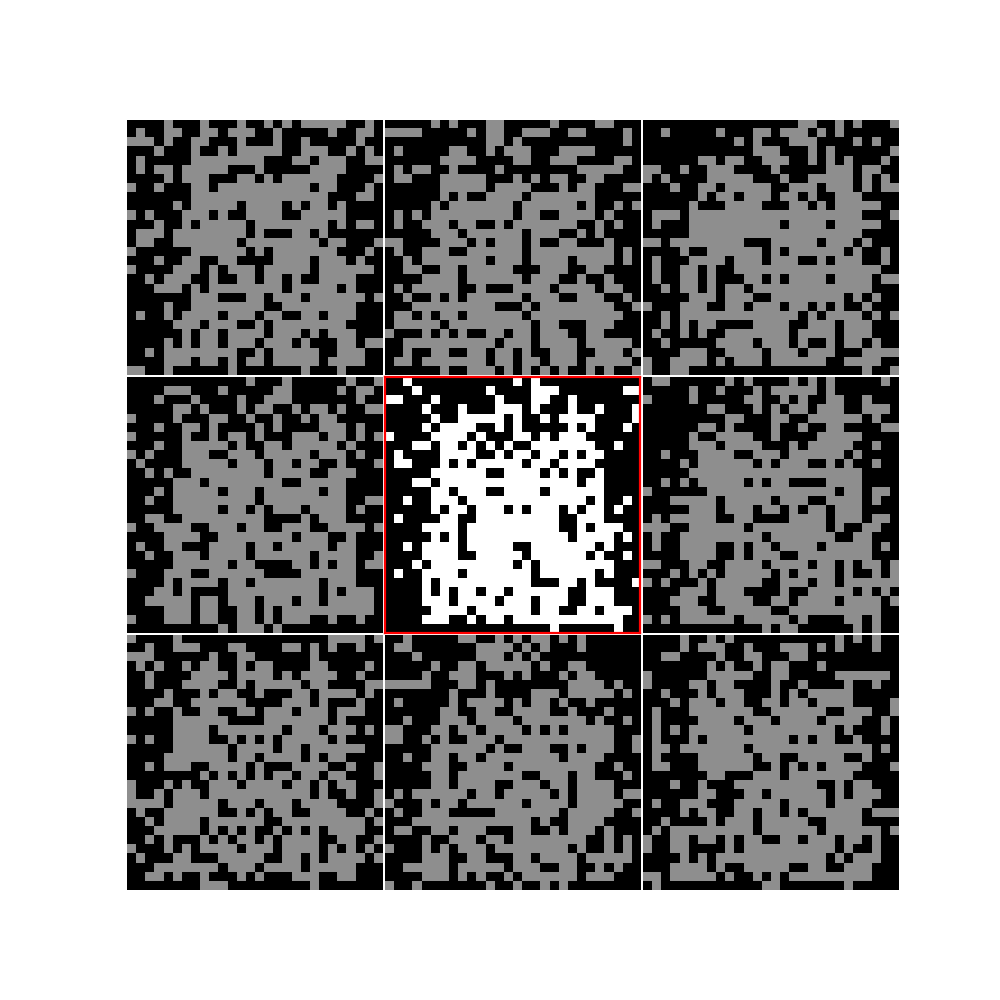}
    \includegraphics[scale=0.42]{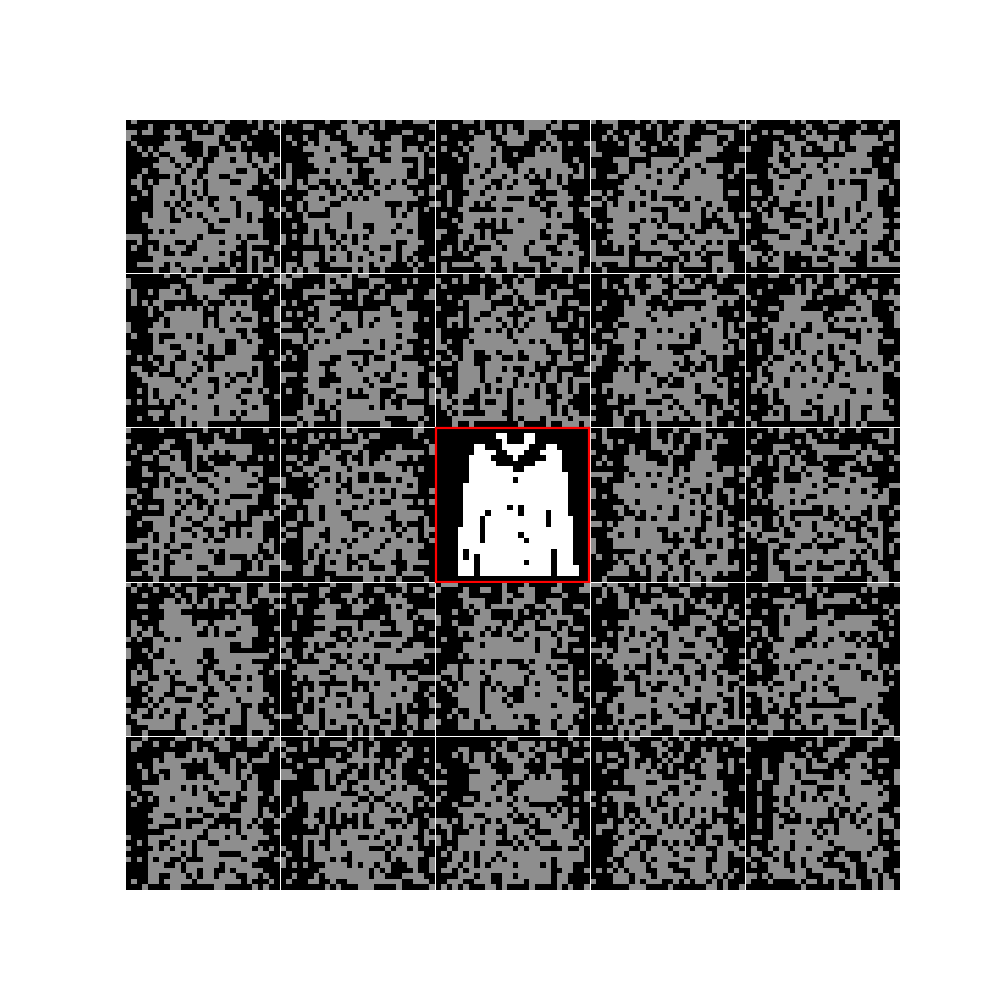}
        \caption{We consider a dense neural network with a degree of interactions $P=4$, and we take just one picture of a coat from the Fashion-MNist dataset \cite{xiao2017fashion}, thresholding the gray-scale values to obtain a binary representation of the image.
    Then, we generate other nine uncorrelated Rademacher archetypes, to reach a amount of patterns $K=10$. Then,
    we produce $M$ noisy examples for each archetype with quality level $r=0.5$ by flipping each pixel with probability as in \eqref{eq:Bernoulli}.
    We focus on the cases where $M=8$ (left) and $M=24$ (right), and show the examples related to the Fashion-MNist coat, for these examples the white pixels are shown in gray. In both plots, they surround the final output of the network, that is the central image in the red box. The original picture of the coat is not shown but differs from the $M=24$ output by just two pixels.
    We notice that, among these two plots, solely in the one with $M=24$ the network is capable to correctly reconstruct the pattern starting from its noisy versions.}
    \label{fig:giacchetto}
\end{figure}

\begin{figure}[t] 
\begin{center}
\begin{tabular}{ccc}
\captionsetup[subfigure]{labelformat=empty}
\subfloat[
]{\includegraphics[width = 0.9in]{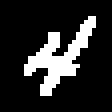}} &
\captionsetup[subfigure]{labelformat=empty}
\subfloat[
]{\includegraphics[width = 0.9in]{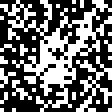}} &
\captionsetup[subfigure]{labelformat=empty}
\subfloat[
]{\includegraphics[width = 0.9in]{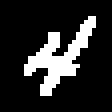}} 
\end{tabular}\\
\vspace{-3mm}
\begin{tabular}{ccc}
\captionsetup[subfigure]{labelformat=empty}
\subfloat[(a)
pattern]{\includegraphics[width = 0.9in]{1a}} &
\captionsetup[subfigure]{labelformat=empty}
\subfloat[(b) 
input \\ configuration]{\includegraphics[width = 0.9in]{1b.png}} &
\captionsetup[subfigure]{labelformat=empty}
\subfloat[(c) 
pattern \\ reconstruction]{\includegraphics[width = 0.9in]{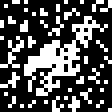}} 
\end{tabular}
\caption{Comparison between the retrieval capabilities exhibited by a dense network ($P=4$, upper line) and by the pair-wise  Hopfield model ($P=2$, lower line). We built both the networks with $N=784$ Ising neurons and we chose a picture from the MNist dataset (i.e., the number $4$). Then, we generated other $35$ independent Rademacher archetypes, in order reach $K=36$; for each archetype the network unsupervisedly experienced $M=80$ examples at a level characterized by a quality $r=0.375$. In the panels in the first column (a) we report the archetype, in the middle column (b) we report a noisy example inputted to the network and in the last column (c) we show the ultimate network's reconstruction where it shines that, while the Hopfield model fails, the dense network correctly performs pattern recognition.}
\label{fig:mnist444}
\end{center}
\end{figure}

As we will see, the answers to the first two points highlight a conflicting role of the interaction order $P$: on the one hand, by increasing $P$ the number of examples required for a sound training grows exponentially ($\propto 1/(P r^{2P})$), on the other hand, the number of storable archetypes also grows exponentially ($\propto N^{P-1}$).
Furthermore, to address the last question, we can introduce a supplementary, additive noise $\omega$ to be applied to the couplings $J_{i_1 i_2 ... i_P}$ that mimics possible flaws occurring during the training and we show that there is an interplay between this kind of noise and the load: if we can afford a downgrade in terms of load (i.e., $b < P-1$), then the network can work even in the presence of extensive noise (i.e., $\omega$ scaling with $N$) on the weights. We anticipate that these features stem from, respectively, the vast available resources ( the $K$ archetypes are allocated in a tensor made of $N^{P}$ elements) and from the redundancy generated when employing over-sized resources.

These concepts are in part visualized in Figs.~\ref{fig:giacchetto}-\ref{fig:mnist444}.
Indeed, in Fig. \ref{fig:giacchetto} 
we show that,
without a sufficient number $M$ of examples, the network is incapable of generalizing an archetype starting from noisy versions of it.
Moreover, it can happen that dense networks are able, from a noisy initial configuration, to recall the reference pattern better than their non-dense counterparts, as evidenced in Fig.~\ref{fig:mnist444}, even if the number of examples given to the network and all the other parameters are equal.

\par\medskip
To make the above statements quantitative, we need a set of observables to assess the retrieval ability of the network, therefore we state the following
\begin{definition}
\label{def:orderparam}
The order parameters of the dense unsupervised Hebbian neural network introduced in Def. \ref{def:dense_unsup} are
\begin{eqnarray}
\label{eq:def_m}
         m_{\mu}&:=&\dfrac{1}{N}\SOMMA{i=1}{N}\xi_i^{\mu}\sigma_i \\
         \label{eq:def_n}
         n_{\mu,a}&:=&\dfrac{r}{\R}\dfrac{1}{N}\SOMMA{i=1}{N}\eta_i^{\mu\,,a}\sigma_i,\\
         \label{eq:def_q}
         q_{lm}&:=&\dfrac{1}{N}\SOMMA{i=1}{N}\sigma_i^{(l)}\sigma_i^{(m)},
\end{eqnarray}
for $\mu=1, \hdots , K$ and $a=1, \hdots , M$. 
\end{definition}

Note that the Mattis magnetization $m_{\mu}$, already defined in \eqref{eq:Mattis}, quantifies the alignment of the network configuration $\boldsymbol \sigma$ with the archetype $\boldsymbol \xi^{\mu}$, $n_{\mu,a}$ quantifies the alignment of the network configuration with the example $\boldsymbol \eta^{\mu,a}$, and $q_{lm}$ is the standard two-replica overlap between the replicas $\boldsymbol \sigma^{(l)}$ and $\boldsymbol \sigma^{(m)}$. 

\section{Cost and Loss functions} \label{sec:loss}
Before proceeding with the investigation of the model, it is worth examining whether the order parameters introduced in the statistical-mechanics context to measure the ability of the system to learn and retrieve archetypes from examples display any connection with the quantities usually employed in the machine-learning field.
There, one typically introduces a {\em loss function} $\mathcal{L}$, namely a positive-definite function that maps any weight setting onto a real number representing some ``cost" associated with that setting; during training the weights are tuned in such a way that $\mathcal{L}$ is lowered and it reaches zero if and only if the network has learnt. Therefore, the goal of the learning stage is to vary the weights  with some algorithm (e.g., contrastive divergence, back-propagation) in order to minimize  $\mathcal{L}$ as the various examples are provided to the network.

In the present case we want the system to learn to reconstruct archetypes from examples and the weights where the learnt information is allocated are the Hebbian couplings $\boldsymbol J$. Following an iterative procedure analogous to those adopted in a machine-learning context, we should prepare the system in some initial configuration $(\boldsymbol \sigma^{(0)}, \boldsymbol J^{(0)})$. Next, we should let the neurons (which are the fast degrees of freedom) relax to some $\boldsymbol \sigma^{(\textrm{eq})}_{\boldsymbol J^{(0)}}$, then evaluate the performance by some $\mathcal{L}^{(0)}:=\mathcal{L}(\boldsymbol \sigma^{(\textrm{eq})}_{\boldsymbol J^{(0)}})$ and modify couplings (e.g., via gradient descent) as $\boldsymbol J^{(0)} \to \boldsymbol J^{(1)}$ in such a way that $\mathcal{L}^{(1)} \leq \mathcal{L}^{(0)}$, and so on so forth up to sufficiently small values of the loss function.
For the present task, focusing on the $\mu$-th pattern, we envisage the following loss function 
\begin{equation}
\mathcal{L}_{\mu}^{(n)}:= \frac{1}{4N^2} \Vert \boldsymbol \xi^{\mu} + \boldsymbol \sigma^{(\textrm{eq})}_{\boldsymbol J^{(n)}} \Vert^2 \cdot \Vert \boldsymbol \xi^{\mu} - \boldsymbol \sigma^{(\textrm{eq})}_{\boldsymbol J^{(n)}} \Vert^2,
\end{equation}
in such a way that $\mathcal{L}_{\mu}^{(n)} \geq 0$, and it reaches zero if and only if the system is retrieving the marked pattern (or its inverse, by gauge invariance).
The loss function defined above can be recast in terms of the Mattis overlaps as 
\begin{equation}
\mathcal{L}_{\mu}^{(n)} = [1 + m_{\mu}^{(n)}] [1 - m^{(n)}_{\mu}],
\end{equation}
where $m_{\mu}^{(n)}:=\frac{1}{N}\boldsymbol \sigma^{(\textrm{eq})}_{\boldsymbol J^{(n)}} \cdot \boldsymbol \xi^{\mu}$,
in such a way that the retrieval region in the phase diagram also highlights the set of values for the control parameters where $\mathcal{L}_{\mu}^{(n)}$ is vanishing.

Further, we notice that the cost function of the the $P$-spin Hopfield model
can be written as
\begin{equation}
\label{eq:H_L}
\mathcal H_{N,K}^{(P)}(\boldsymbol \sigma | \boldsymbol \xi)=-\frac{N}{P!}\sum_{\mu=1}^{K} m_{\mu}^P = -\frac{N}{P!}\sum_{\mu=1}^{K}\left(1-\mathcal L_{\mu}^{*}\right)^{\frac{P}{2}},
\end{equation}
where $\mathcal L_{\mu}^{*}$ is the loss function evaluated for the generic configuration $\boldsymbol \sigma$. We should also keep in mind that we are considering a learning process where the available dataset is $\{\boldsymbol \eta^{\mu,a} \}_{\mu=1,...,K}^{a=1,...,M}$ with $\mu$ undisclosed. The most natural way to recover that framework is by replacing, in \eqref{eq:H_L}, archetypes with examples and then averaging over the latter; by doing so we recover the Hamiltonian \eqref{def:H_PHopEx}.

Now, in the noiseless limit $\beta \to \infty$, the system spontaneously relaxes to configurations corresponding to the lowest energy which ensure that $\mathcal L_{\mu}^{*}=0$.
In order to see that only one of the losses that are summed up in the previous equation is minimised, and that
the network will not attempt to minimise each of them at the same
time, we can evaluate the average energy for the whole class of possible
retrieval states. The most probable candidate states for retrieval
are given by linear combinations of $n$-patterns:
\begin{equation} \label{eq:configu_spu}
\sigma_{i}=\mathrm{sign}\left(\sum_{k=1}^{n}\xi_{i}^{\mu_k}\right),
\end{equation}
their  Mattis overlap is
%
\begin{equation}
m_{\mu_{\ell}}=\frac{1}{N}\sum_{i=1}^{N}\xi_{i}^{\mu_{\ell}}
\mathrm{sign}\left(\xi_{i}^{\mu_{\ell}}+\sum_{\substack{k=1 \\ k \neq \ell}}^{n}\xi_{i}^{\mu_k}\right).
\end{equation}
Averaging over pattern realization, for any $k \leq n$ we have
\begin{equation}
\mathbb{E}_{\bm\xi}[m_{\mu_k}]=\int\frac{dz}{\sqrt{2\pi}}\exp\left(-\frac{z^{2}}{2}\right)\mathrm{sign}\left(1+\sqrt{n-1}z\right)=\mathrm{erf}\left(\frac{1}{\sqrt{2(n-1)}}\right),
\end{equation}
while $\mathbb{E}_{\bm\xi}[m_{\mu_k}]=0$ for any $k>n$.
We now estimate the expected energy for configurations like \eqref{eq:configu_spu} obtaining:
\begin{equation} 
 \mathcal H_{N,K,M}^{(P)}(\boldsymbol \sigma | \boldsymbol \xi)=-\frac{nN}{P!}\left[\mathrm{erf}\left(\frac{1}{\sqrt{2(n-1)}}\right)\right]^{P}
\end{equation}
Since this expression is a strictly increasing function of $n$,
the only stable retrieval states are those with $n=1$, thus their
Mattis overlap will tend to $1$, showing indeed that the network minimised
only one of the $\mathcal L^*_{\mu}$ losses at a time.

\section{Analytical findings} \label{sec:solution}
In this section we solve the dense unsupervised Hopfield model introduced in Definition \ref{def:dense_unsup}, specifically, we obtain an explicit expression for its  quenched statistical pressure (i.e., the free energy) in terms of the order parameters of the theory.  Then, by extremizing the free energy w.r.t. these order parameters we obtain a set of self-consistency equations for the latter whose inspection allows us to obtain the phase diagrams of the model. To this aim we exploit Guerra's interpolation technique \cite{guerra_broken} which allows us to get the free-energy explicitly. However, as we will see,
some adjustments to the standard protocol are in order in the estimate of the noise distribution, which, unlike pairwise networks, can not be directly considered as a Gaussian random variable. The core problem is that the distributions of the post-synaptic potentials are not Gaussian here, hence standard universality of spin-glass noise \cite{CarmonaWu, Genovese, Longo} does not apply straightforwardly. 
To overcome this obstacle, we apply the Central Limit Theorem (CLT) in order to estimate it as a single Gaussian variable.
\newline
Before proceeding, we highlight that, as standard (see e.g., \cite{Coolen}) and with no loss of generality, in the following we will focus on the ability of the network to learn and retrieve the first archetype $\boldsymbol \xi^1$. Thus, in the next expression, the contribution corresponding to $\mu=1$ shall be split from all the others and interpreted as the {\em signal} contribution, while the remaining ones make up the slow-noise contribution impairing both learning and retrieval of $\boldsymbol \xi^1$, namely starting from Eq. \eqref{eq:Def_orig_Z}, we apply the functional-generator technique\footnote{We add the term $J \sum_i \xi_i^1 \si$ to generate the expectation of the Mattis magnetization $m_1$: the latter emerges by evaluating the derivative w.r.t. $J$ of the quenched statistical pressure at $J=0$.} to get
\begin{align} \label{eq:part_S_N}
    \mathcal{Z}^{(P)}_{N,K,M ,r,\beta} (\bm\eta) =& \lim_{J \rightarrow 0} \mathcal{Z}^{(P)}_{N,K,M,r,\beta}(\bm\xi^1,\bm\eta;J) \notag \\
    =& \lim_{J \rightarrow 0} \sum_{\bm \sigma}  \exp \left[ J \sum_{i=1}^N \xi_i^1 \si + \dfrac{\beta '}{2\,\R^{P/2} M\,N^{P-1}}\SOMMA{a=1}{M}\left(\sum_{i=1}^N\eta_{_{i}}^{a,1}\sigma_{_{i}}\right)^{^P} \right. \notag 
    \\
    &\left.+\dfrac{\beta 'P!}{2\R^{P/2} N^{P-1}} \, \SOMMA{\mu>1}{K} \SOMMA{ (i_{_1},\cdots ,i_{_{P}})}{N,\cdots,N}\left(\dfrac{1}{M}\SOMMA{a=1}{M} \eta^{\mu\,,a}_{i_1}\cdots\eta^{\mu\,,a}_{i_{P}}\right)\sigma_{_{i_1}}\cdots\sigma_{_{i_{_{P}}}}\right].
\end{align}

{Focusing only on the noise terms in round brackets in \eqref{eq:part_S_N}
we can apply the CLT and approximate it with a Gaussian variable with suitable first and second momenta. Therefore we can recast this term as follows   
\begin{equation}
    \label{eq:con_V}
    \left(\dfrac{1}{M}\SOMMA{a=1}{M}\eta^{\mu\,,a}_{i_1}\cdots\eta^{\mu\,,a}_{i_{P}}\right)\sim r^P\sqrt{1+\rho_P}\lambda^\mu_{i_1,\hdots,i_P}\;\;\;\mathrm{with}\;\;\;\lambda^\mu_{i_1,\hdots,i_P}\sim\mathcal{N}(0,1)
\end{equation}
where
$\rho_P=\dfrac{1-r^{2P}}{Mr^{2P}}$.
Remarkably, this reasoning shows that these dense networks exhibit the universality of the quenched noise \cite{CarmonaWu,Genovese,Longo}, namely we can approximate the overall field experienced by a neuron (i.e., the post-synaptic potential) as a random Gaussian field.
}

Now, plugging \eqref{eq:con_V} into \eqref{eq:part_S_N} we reach a useful expression for the partition function for the unsupervised dense Hebbian neural network as
\begin{align} \label{eq:vera_Z}
    \mathcal{Z}^{(P)}_{N,K,M, r,\beta} (\bm\xi^1,\bm\eta; J) =&  \sum_{\bm \sigma}  \exp \left[ J \sum_{i=1}^N \xi_i^1 \si + \dfrac{\beta '}{2\R^{P/2}  M\,N^{P-1}}\SOMMA{a=1}{M}n_{1,a}^{^P} \right. \notag \\
    &\left.+\dfrac{\beta 'P!\sqrt{1+\rho_P}}{2(1+\rho)^{P/2} \,N^{P-1}}\SOMMA{\mu>1}{K}\left(\SOMMA{(i_{_1},\cdots, i_{_{P}})}{N,\cdots,N}\lambda^{\mu}_{i_1,\hdots,i_P}\:\sigma_{_{i_1}}\cdots\sigma_{_{i_{_{P}}}}\right)\right]
\end{align}
where we exploit the relation $\R = r^2(1+\rho)$.

\par\medskip
We now proceed by applying Guerra's interpolation.
The underlying idea behind this technique is to introduce a generalized free-energy which interpolates between the original one (which is the target of our investigation but we are not able to address it directly) and a simple one (which we can solve exactly). The latter is typically a one-body model mimicking the original one: the fields acting on neurons are chosen to exhibit statistical properties that simulate those experienced by neurons in the original model and due to the effect of other neurons.
We thus find the solution of the simple model and we propagate the obtained solution back to the original model by the fundamental theorem of calculus. In this last passage we assume RS (\emph {vide infra}), namely, we assume that the order-parameter fluctuations are negligible in the thermodynamic limit: this property makes the integral in the fundamental theorem of calculus analytical. Let us proceed by steps and give the next definitions 
\begin{definition}
\label{def:part_Interpolante_RS}
Given the interpolating parameter $t \in [0,1]$, the constants $A, \ \psi \in \mathbb{R}$ to be set a posteriori, and the i.i.d. standard Gaussian variables $Y_i\sim \mathcal{N}(0,1)$ for $i=1, \hdots , N$, the interpolating partition function is given as 

\begin{equation}
\begin{array}{lll}
     \mathcal{Z}^{(P)}_{N, K, M, r,\beta}&(\boldsymbol{\xi}^1,\boldsymbol{\eta}; J, t)  \coloneqq  \SOMMA{\boldsymbol \sigma}{}   \mathcal B_{N,K, M, \beta}^{(P)} 
 (\boldsymbol{\sigma} |\boldsymbol{\xi}^1, \boldsymbol{\eta}; J, t ).
     \label{def:partfunct_GuerraRS}
\end{array}
\end{equation}
%
%
where $B_{N,K, M, r, \beta}^{(P)} $ is the related Boltzmann factor reads as
\begin{equation}
\begin{array}{lll}
\mathcal B^{(P)}_{N,K, M, r,\beta} &
 (\boldsymbol{\sigma} |\bm\xi^1,\boldsymbol{\eta}; J, t )
	\coloneqq\exp{\Bigg[}J \sum_{i=1}^N \xi_i^1 \si+\dfrac{t \beta 'N}{2M}\left(1+\rho \right)^{P/2}\SOMMA{a=1}{M}n_{1,a}^{^P} +\psi(1-t)N\SOMMA{a=1}{M}n_{1,a}
        \\\\
        &+\sqrt{t}\dfrac{\beta 'P!\sqrt{1+\rho_P}}{2(1+\rho)^{P/2} N^{P-1}}\SOMMA{\mu>1}{K}\SOMMA{(i_{_1},\cdots, i_{_{P}})}{N,\cdots,N}\lambda^{\mu}_{i_1,\hdots, i_P}\sigma_{_{i_1}}\cdots\sigma_{_{i_{_{P}}}}+\sqrt{1-t}A\SOMMA{i=1}{N}Y_i\sigma_{i}\Bigg] ;
	\end{array}
	\end{equation}
A generalized average follows from this generalized measure as
\beq
	\omega_{t} [(\cdot)] \coloneqq  \frac{1}{\mathcal{Z}^{(P)}_{N, K, M, r,\beta}(\boldsymbol{\xi}^1,\boldsymbol{\eta}; J, t) } \, \sum_{\boldsymbol \sigma} ~ (\cdot) ~   \mathcal B^{(P)}_{N,K, M, r,\beta} 
 (\boldsymbol{\sigma} |\bm\xi^1,\boldsymbol{\eta}; J, t )
	\eeq
	and
\beq
\langle (\cdot)   \rangle_{t}  \coloneqq \mathbb E \{ \omega_{t} [( \cdot)] \},
\eeq
where the expectation $\mathbb E$ is now  meant  over any $\lambda^\mu_{i_1, \hdots , i_{P}}$ and $Y_i$ too.
\newline

The interpolating quenched statistical pressure related to the partition function (\ref{def:partfunct_GuerraRS}) is introduced as
\begin{equation}
\mathcal{A}^{(P)}_{N,K,M,r ,\beta}(J, t) \coloneqq \frac{1}{N} \mathbb{E} \left[  \ln \mathcal{Z}^{(P)}_{N,K, M ,r,\beta}(\boldsymbol{\xi}, \boldsymbol{\eta}; J, t)  \right],
\label{hop_GuerraAction}
\end{equation}
%
and, in the thermodynamic limit,
\begin{equation}
\mathcal{A}_{\alpha_{b},M, r, \beta}^{(P)}(J, t) \coloneqq \lim_{N \to \infty} \mathcal{A}^{(P)}_{N, K, M , r, \beta}(J,t).
\label{hop_GuerraAction_TDL}
\end{equation}

Of course, by setting $t=1$ we recover the original model: the interpolating pressure recovers the original one (\ref{PressureDef_unsup}), that is $\mathcal A_{N,K,M, r, \beta }^{(P)} (J) = \mathcal{A}^{(P)}_{N,K, M, r, \beta}(J, t=1 )$, and analogously for the partition function, the standard Boltzmann measure and the related averages.
\end{definition}

As anticipated, the following analytical results are obtained under the RS hypothesis, namely assuming that, in the thermodynamic limit, the distribution of the generic order parameter $X$ is centered at its expectation value w.r.t. the interpolating measure $\bar X := \langle X \rangle_t$ with vanishing fluctuations for all t, that is,  
\begin{equation} \label{eq:RS}
\lim_{N \to \infty}  \langle (X- \bar{X}) \rangle_t =0.
\end{equation} 
%
Although this assumption is not fulfilled by this kind of systems (at least not everywhere in the space of control parameters), it is usually adopted as it yields only small quantitative corrections and, further, a full replica-symmetry-breaking theory for these systems is still under construction (see e.g., \cite{Crisanti-RSB,Steffan-RSB,AABO-JPA2020,AAAF-JPA2021,Albanese2021}).

\par \medskip


We now proceed to determine the self-consistency equations for the order parameters by extremizing the quenched statistical pressure; to this aim it is mathematically convenient to take the thermodynamic limit and split the discussion in two cases: the high-storage regime $b=P-1$ (corresponding to the highest load allowed \cite{Baldi,Bovier,AFMJMP2022}) and the low-storage regime $b<P-1$; as stressed above, in both cases, we shall consider only even values of $P$ and, specifically, $P \geq 4$\footnote{The case $P=2$ corresponds to the unsupervised Hopfield model treated in \cite{EmergencySN}; the assumption $P \geq 4$ is used in the proof of Theorem \ref{P_quenched} in Appendix \ref{app:proofP-1}.}.

\subsection{High-load regime}\label{ultrastorage}
In this subsection we present the main analytical result obtained in the case where $K/N^{P-1}$ remains finite in the thermodynamic limit, that is $\alpha_{P-1}$ is finite and non-vanishing, see \eqref{eq:carico_TDL}. 
\begin{proposition}
\label{P_quenched}
In the thermodynamic limit ($N\to\infty$), under the RS assumption \eqref{eq:RS}, the quenched statistical pressure for the unsupervised, dense neural-network described by \eqref{eq:vera_Z} set in the high-load regime $b = P-1$ reads as
\begin{equation}
\label{eq:pressure_GuerraRS}
\begin{array}{lll}
     \mathcal{A}_{\gamma,M, r, \beta}^{(P)}(J) 
     =& \mathbb{E}\left \{ \ln{2\cosh{\left[J\xi^1+\beta '\dfrac{P}{2}\n^{P-1} (1+\rho)^{P/2-1}\hat{\eta}+Y \sqrt{\gamma\dfrac{\b \,^2(1+\rho_P)}{(1+\rho)^{P}} \dfrac{P}{2} \q^{^{P-1}}}\right]}}\right\}
         \\\\
         &-\dfrac{\beta '}{2}(P-1)(1+\rho)^{P/2}\n^P+\gamma\dfrac{\b\,^2(1+\rho_P)}{4 (1+\rho)^{P}}(1-P\q^{P-1}+(P-1)\q^P).
\end{array}
\end{equation}
with $\mathbb{E}= \mathbb{E}_\xi\mathbb{E}_{(\eta|\xi)}\mathbb{E}_Y$, $\etaM:=\frac{1}{rM}\SOMMA{a=1}{M}\eta^{1,a}$ and 
$\bar n$ and $\bar q$ fulfill the following self-consistency equations 
\begin{equation}
    \begin{array}{lll}
         \n=\dfrac{1}{1+\rho}\mathbb{E}\left\{\tanh{\left[\beta '\dfrac{P}{2}\n^{P-1} (1+\rho)^{P/2-1}\hat{\eta}+Y \sqrt{\gamma\dfrac{\b \,^2(1+\rho_P)}{(1+\rho)^{P}} \dfrac{P}{2} \q^{^{P-1}}}\right]}\etaM\right\},
         \\\\
         \q=\mathbb{E}\left\{\tanh^{\2}{\left[\beta '\dfrac{P}{2}\n^{P-1} (1+\rho)^{P/2-1}\hat{\eta}+Y \sqrt{\gamma\dfrac{\b \,^2(1+\rho_P)}{(1+\rho)^{P}} \dfrac{P}{2} \q^{^{P-1}}}\right]}\right\}.

    \end{array}
    \label{eq:High_store_self_n_q}
\end{equation}
 
Furthermore, considering the auxiliary field $J$ linked to $\bar{m}$ as $\bar{m}= \nabla_J \mathcal{A}_{\gamma,M, r, \beta}^{(P)}(J) |_{J=0}$, we have
\begin{align}
\label{eq:High_store_self_m}
    \m=\mathbb{E}\left \{\tanh{\left[\beta '\dfrac{P}{2}\n^{P-1} (1+\rho)^{P/2-1}\hat{\eta}+Y \sqrt{\gamma\dfrac{\b \,^2(1+\rho_P)}{(1+\rho)^{P}} \dfrac{P}{2} \q^{^{P-1}}}\right]}\xi^1\right\}.
    \end{align}
\end{proposition}

For the proof we refer to Appendix \ref{app:proofP-1}.

\subsection{Low-load regime}
In this subsection we present the main analytical finding obtained by setting $b<P-1$ in \eqref{eq:carico_TDL}.

\begin{proposition}
\label{P_quenched_aminP-1}
In the thermodynamic limit ($N \to \infty$), under the RS assumption \eqref{eq:RS}, the quenched statistical pressure for the unsupervised, dense neural-network described by \eqref{eq:vera_Z} set in the low-load regime $b < P-1$ reads as
\begin{equation}
\small
\label{eq:pressure_GuerraRS_aminP-1}
\begin{array}{lll}
    \mathcal{A}_{0, M, \beta, r}^{(P)}(J)  =& \mathbb{E}\left \{ \ln{2\cosh{\left[J\,\xi^1+\beta '\dfrac{P}{2}\n^{P-1}(1+\rho)^{P/2-1} \etaM\right]}}\right \}-\dfrac{\beta '}{2}(P-1)(1+\rho)^{P/2}\n^P.
\end{array}
\end{equation}
with $\mathbb{E}=\mathbb{E}_\xi\mathbb{E}_{(\eta|\xi)}$ and $\bar n$ fulfills the following self-consistency equation
\begin{equation}
    \begin{array}{lll}
         \n=\dfrac{1}{1+\rho}\mathbb{E}\left\{\tanh{\left[\beta '\dfrac{P}{2}\n^{P-1}(1+\rho)^{P/2-1} \etaM\right]}\etaM\right\}.
    \end{array}
    \label{eq:self_aminP-1}
\end{equation}
	%
Furthermore, considering the auxiliary field $J$ linked to $\bar{m}$ as $\bar{m}= \nabla_J \mathcal{A}_{0,M, \beta, r}^{(P)}(J) |_{J=0}$ we have
\begin{align}
 \label{eq:m_self_aminP-1}
    \m=\mathbb{E}\left\{\tanh{\left[\beta '\dfrac{P}{2}\n^{P-1}(1+\rho)^{P/2-1} \etaM\right]}\xi^1\right\}.
\end{align}
\end{proposition}

The proof is analogous to the case $b =P-1$ (see Appendix \ref{app:proofP-1}), therefore we will  omit it. 

\par \medskip

Note that, in this low-load regime, we are left with one single order parameter that measures the degree of order in the system, much as like in the Curie-Weiss model \cite{Barra0,Jean0}. In fact, here $\bar q =0$ and this indicates a loss of slow-noise or of ``glassiness'' in the network, analogously to what happens for the pairwise Hopfield model in the low-storage regime $\lim_{N \to \infty}K/N=0$. However, in this dense generalization, this regime deserves a particular attention because, as we will see in Sec.~\ref{sec:ultra-noise}, a relatively small load can release some resources to handle possible supplementary noise due, for instance, to flaws underlying storing \cite{AgliariDeMarzo}.  
We anticipate that, in that case,
we ultimately recover self-consistence equations similar to those obtained in high-load ($b=P-1$), but the noise term, instead of stemming from the load, will be related to this additional disturbance. 

\subsection{Additive noise in low-load regime} \label{sec:ultra-noise}

As mentioned in the previous section and shown numerically in the next one, by increasing the interaction order $P$ among neurons, the storage capacity increases arbitrarily ($K\sim N^{P-1}$ -- \textit{high-load regime}). However, our model assumes that the coupling tensor $\boldsymbol J$ is devoid of flaws, whereas, in general, the communication among neurons can be disturbed hence affecting the synaptic processes (e.g., see \cite{BarraPRLdetective,Battista2020}). It is then natural to question if unsupervised dense neural networks described by \eqref{eq:vera_Z} are robust versus this kind of noise too.

   Recalling the Hamiltonian \eqref{def:H_PHopEx}  we can write
   \begin{align}
    \mathcal H^{(P)}_{N,K,M}(\bm\sigma|  \bm \eta) =
    &- \SOMMA{(i_1,\hdots, i_P)}{N,\hdots,N}J_{i_{1}\cdots i_{P}} \sigma_{i_{1}}\cdots\sigma_{i_{P}}  \notag\\
    &=-\dfrac{1}{\R^{P/2}\,M\,N^{P-1}}\SOMMA{\mu=1}{K}\SOMMA{a=1}{M}\SOMMA{(i_1,\hdots, i_P)}{N,\hdots,N}\eta^{\mu,a}_{i_{1}}\cdots \eta^{\mu,a}_{i_{P}} \sigma_{i_{1}}\cdots\sigma_{i_{P}},
    \end{align}
    where we outlined the entry of the coupling tensor $\bm J$. 
   Then, following \cite{AgliariDeMarzo}, we model the supplementary noise, by introducing an additional, random contribution as
   \begin{equation} \label{eq:J_noise}
   J_{i_1\cdots i_P} \to  \tilde{J}_{i_1\cdots i_P} = \eta^{\mu,a}_{i_1}\cdots\eta^{\mu,a}_{i_P}+w \,\tilde\eta_{i_1\cdots i_P}^{\mu,a},\\
   \end{equation}
with $w \in \mathbb R$ and $\tilde\eta_{i_1\cdots i_P}^{\mu,a} \sim_{\textrm{iid}} \mathcal N(0,1)$.

We investigate the effects of such a noise on the retrieval capabilities of the system and the existence of upper bounds for the amount of noise that the system can tolerate without loosing its ability to play as an associative memory.
 
As shown in \cite{AgliariDeMarzo}, if we can afford a downgrade in terms of load (i.e. $b < P-1$),
we can consider the presence of extensive synaptic noise that grows algebraically with the network size $N$, namely 
\begin{align}
\label{eq:noise}
    w=\tau\,N^{\delta}, ~ \textrm{with} ~~ \tau \in \mathbb R ~~ \textrm{and} ~~ \delta \in \mathbb R^+.
\end{align}

In fact, following the same path presented in the first part of this section, including the noise defined in \eqref{eq:J_noise} and \eqref{eq:noise} yields self-consistency equations for the order parameters that display the same expression found in the high-storage regime \eqref{eq:self_aminP-1}-\eqref{eq:m_self_aminP-1}, as long as we replace $\beta '$ with $\tau\beta '$ in the noise part, namely 
\begin{equation}
\begin{array}{lll}
         \n=\dfrac{1}{1+\rho}\mathbb{E}\left\{\tanh{\left[\beta '\dfrac{P}{2}\n^{P-1} (1+\rho)^{P/2-1}\etaM+Y \beta '\tau\sqrt{\gamma\dfrac{(1+\rho_P)}{(1+\rho)^{P}} \dfrac{P}{2} \q^{^{P-1}}}\right]\etaM}\right\},
         \\\\
         \q=\mathbb{E}\left\{\tanh{}^{\2}{\left[\beta '\dfrac{P}{2}\n^{P-1} (1+\rho)^{P/2-1}\etaM+Y \beta '\tau\sqrt{\gamma\dfrac{(1+\rho_P)}{(1+\rho)^{P}} \dfrac{P}{2} \q^{^{P-1}}}\right]}\right\},
         \\\\
         \m=\mathbb{E}\left\{\tanh{\left[\beta '\dfrac{P}{2}\n^{P-1} (1+\rho)^{P/2-1}\etaM+Y \beta '\tau\sqrt{\gamma\dfrac{(1+\rho_P)}{(1+\rho)^{P}} \dfrac{P}{2} \q^{^{P-1}}}\right]}\xi^1\right\}.
    \end{array}
\end{equation}

Then, one can show that, if we have an extensive noise \eqref{eq:noise} with $\delta < \dfrac{P-1-b}{2}$, the noise contribution in the hyperbolic tangent in the previous equations is vanishing and we recover the low-load scenario. In other words, there is an interplay between the load (ruled by $b$), the interaction order (ruled by $P$) and the supplementary noise (ruled by $\delta$). Thus, when one of these is enhanced, the others must be overall suitably downsized if we want to preserve the retrieval capability of the system.

\par \medskip
Before concluding we stress that the results obtained in this subsection are not influenced by the dataset parameters $M$ and $r$; this implies that we can not leverage either the quality or the quantity of the dataset to mitigate the effects of this supplementary noise.

\subsection{Low-entropy datasets in the high-load regime}
\label{sec:datalimit_unsup}
\label{betainftyunsup}\label{ssec:ground}

As explained in Sec.~\ref{definitions}, the parameter $\rho = (1-r^2)/(Mr^2)$ quantifies the amount of information needed to describe the original message
$\boldsymbol\xi^\mu$ given the set of related examples $\{\boldsymbol\eta^{\mu,a}\}^{a=1,...,M}$. In this section we focus on the case $\rho \ll 1$ that corresponds to a low-entropy dataset or, otherwise stated, to a high-informative dataset.
The advantage of this analysis is that, under this condition, we obtain a relation between $\bar n$ (a natural order parameter of the model) and $\bar m$ (a practical order parameter of the model)\footnote{It is worth recalling that the model is supplied only with examples -- upon which $\{n^{\mu,a}\}$ are defined -- while it is not aware of archetypes  -- upon which $\{m^{\mu}\}$ are defined. The former constitute natural order parameters and, in fact, the Hamiltonian $\mathcal H^{(P)}_{N,K,M, r}$ in \eqref{def:H_PHopEx} can be written in terms of the example overlaps. The latter are practical order parameters through which we can assess the capabilities of the network.}, thus, the self-consistency equation for $\bar n$ can be recast into a self-consistency equation for $\bar m$ and its numerical solution versus the control parameters allows us to get the phase diagram for the system more straightforwardly. 

As explained in Appendix \ref{app:proofP-1}, we start from the self-consistency equations found in the high-storage regime \eqref{eq:High_store_self_n_q}-\eqref{eq:High_store_self_m} and we exploit the CLT to write $\etaM\sim 1+\lambda\sqrt{\rho}$. In this way we reach the simpler expressions
\begin{eqnarray}
    (1+\rho)\n&=&\m +\beta ' \dfrac{P}{2}\rho\left(1+\rho\right)^{P/2-1}(1-\q)\n^{^{P-1}},\label{eq:n_Mgrande}
    \\
    \q&=&\mathbb{E}_{_Z}\left[\tanh{}^{\2}{ g(\beta, Z, \bar{n})}\right],\label{eq:n_of_M_unsup}
    \\
    \m&=& \mathbb{E}_{_Z}\left[\tanh{ g(\beta, Z, \bar{n})}\right],\label{eq:m_Mgrande}
\end{eqnarray}
where 
 \begin{align}
    &g(\beta, Z, \bar{n})=\beta '\dfrac{P}{2}\n^{^{P-1}}(1+\rho)^{P/2-1}+\beta 'Z\sqrt{\rho\dfrac{P^2}{4}\n^{^{2P-2}}(1+\rho)^{P-2}+\gamma\dfrac{(1+\rho_P)}{(1+\rho)^{P}} \dfrac{P}{2} \q^{^{P-1}}\;}\;
    \label{eq:g_of_unsuper_n}
 \end{align}
and $Z \sim \mathcal{N}(0,1)$ is a standard Gaussian variable. \\

Focusing on the argument of the hyperbolic tangent \eqref{eq:g_of_unsuper_n}, we can split it into three parts: the first one represents the amplification of the signal; the second one reflects the use of perturbed version of the retrieved pattern and not the pattern themselves; the third one is the noise linked to the presence of the other patterns. 

Further, in the retrieval region, where $1-\q$ is vanishing, as long as $\rho \ll 1$, we can truncate the right-hand-side of \eqref{eq:n_Mgrande} into $\n (1+\rho)\sim \m$.
This leads to significant advantages in the computation time required to get a numerical solution of the self-consistency equations.
In fact, by using $\n(1+\rho) = \m$, in the argument of hyperbolic tangent
, we get
\begin{align}
    &g(\beta, Z, \bar{m})=\Tilde{\beta}\dfrac{P}{2}\m^{^{P-1}}+\Tilde{\beta} Z\sqrt{\rho\left(\dfrac{P}{2}\m^{^{P-1}}\right)^2+\gamma (1+\rho_P) \dfrac{P}{2} \q^{^{P-1}}\;},
    \label{eq:g_of_unsuper_m_mod}
\end{align}
%
where we put
\begin{equation} \label{eq:tildebeta}
\tilde{\beta} = \frac{\beta'}{(1+\rho)^{\frac{P}{2}}} =  
\frac{2 \beta}{P!} \frac{1}{(1+\rho)^{\frac{P}{2}}}.
\end{equation}
The consequent, remarkable reward of this truncation consists in retaining only two of the three self-consistency equations, namely only the ones for $\q$ and $\m$, while the resulting error by this truncation is numerically small, as checked in Fig.~\ref{fig:truncated} where we plot $\n$ versus $r$ for different values of the parameters and compare the outcomes obtained with and without the truncation. 

\begin{remark}
     We stress that here we are focusing only on the low entropy dataset limit because in the \textit{high entropy scenario}, i.e. $\rho \gg 1$, the \eqref{eq:n_Mgrande} will reduce to
    \begin{equation}
        \n = \tilde\beta \dfrac{P}{2}\rho(1+\rho)^{P-2}(1-\q)\n^{P-1},
    \end{equation}
    whose solutions are $\n=0$ or
    \begin{equation}
        \n = \dfrac{1}{(1+\rho)}\left(\tilde\beta \dfrac{P}{2}\rho(1-\q)\right)^{-\frac{1}{P-2}}\,.
    \end{equation}
    Replacing this expression in \eqref{eq:g_of_unsuper_n} we get a signal term that reads as 
    \begin{align}
    \tilde\beta\dfrac{P}{2}(1+\rho)^{P-1}\n^{P-1}\underset{\rho\gg1}{\sim} \rho^{-\frac{P-1}{P-2}}.
    \end{align}
    Therefore, the signal term in \eqref{eq:g_of_unsuper_n} will be strongly suppressed and the retrieval process will no longer be possible. 
\end{remark}

\begin{figure}[t]
    \centering
    \includegraphics[width=15cm]{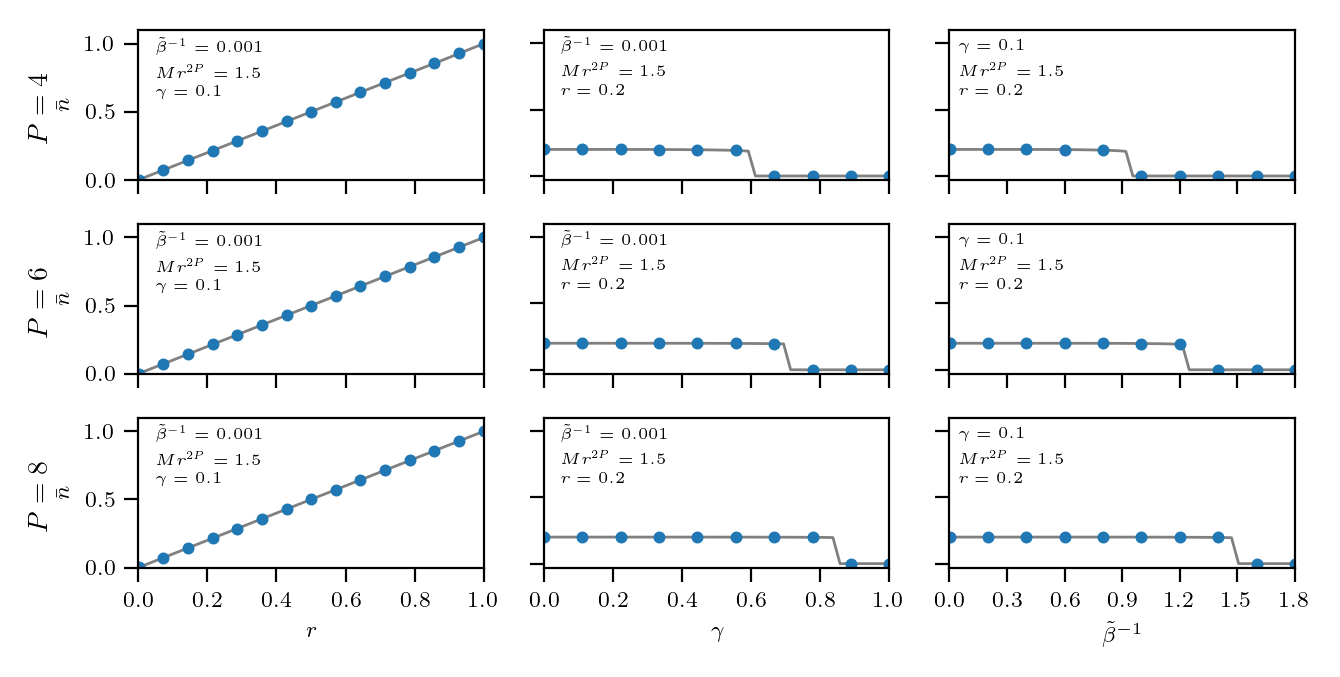}
    \caption{We compute $\n$ solving numerically \eqref{eq:n_Mgrande} for different values of $P$, $\tilde\beta$, $\gamma$ and $M$, as reported in each panel, and we plot it versus $r$. We compare the results obtained using the exact expression of $\n$ (blue dots) and those obtained using the approximated one (namely $\n(1+\rho)=\m$, solid grey line).  We notice that there is completely agreement between the exact expression and the approximated one, whatever the value of $P$ considered.} 
    \label{fig:truncated}
\end{figure}

\medskip
We now further handle the Eqs.~\eqref{eq:n_Mgrande} by computing their zero-temperature limit. 
As detailed in the Appendix \ref{app:proofP-1}, by taking the limit $\beta \to \infty$ in Eqs.~\eqref{eq:n_of_M_unsup} and \eqref{eq:m_Mgrande} we get
\begin{equation}\label{GroundZero}
    \begin{array}{lll}
         \m = \mathrm{erf}\left(\dfrac{P}{2}\dfrac{\m^{P-1}}{G}\right),\;\;\;\;\q=1,
          \\\\
         G=\sqrt{2\left[\rho\left(\dfrac{P}{2}\m^{^{P-1}}\right)^2+\gamma(1+\rho_P) \dfrac{P}{2} \right]\,\;}.
    \end{array}
\end{equation}

\section{Numerical findings}

\begin{figure}[t]
    \centering
    \includegraphics[scale=0.4]{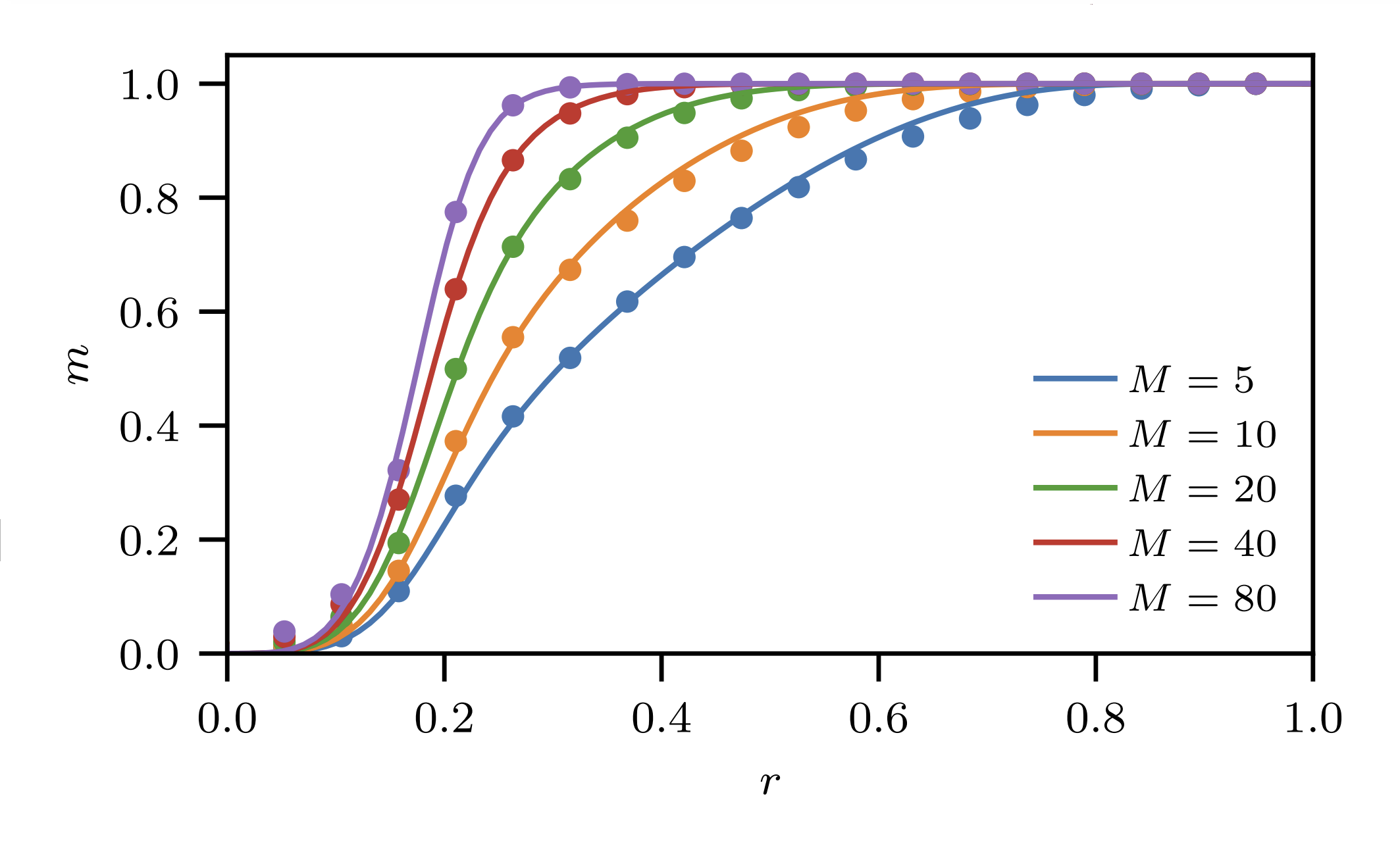}
    \includegraphics[scale=0.4]{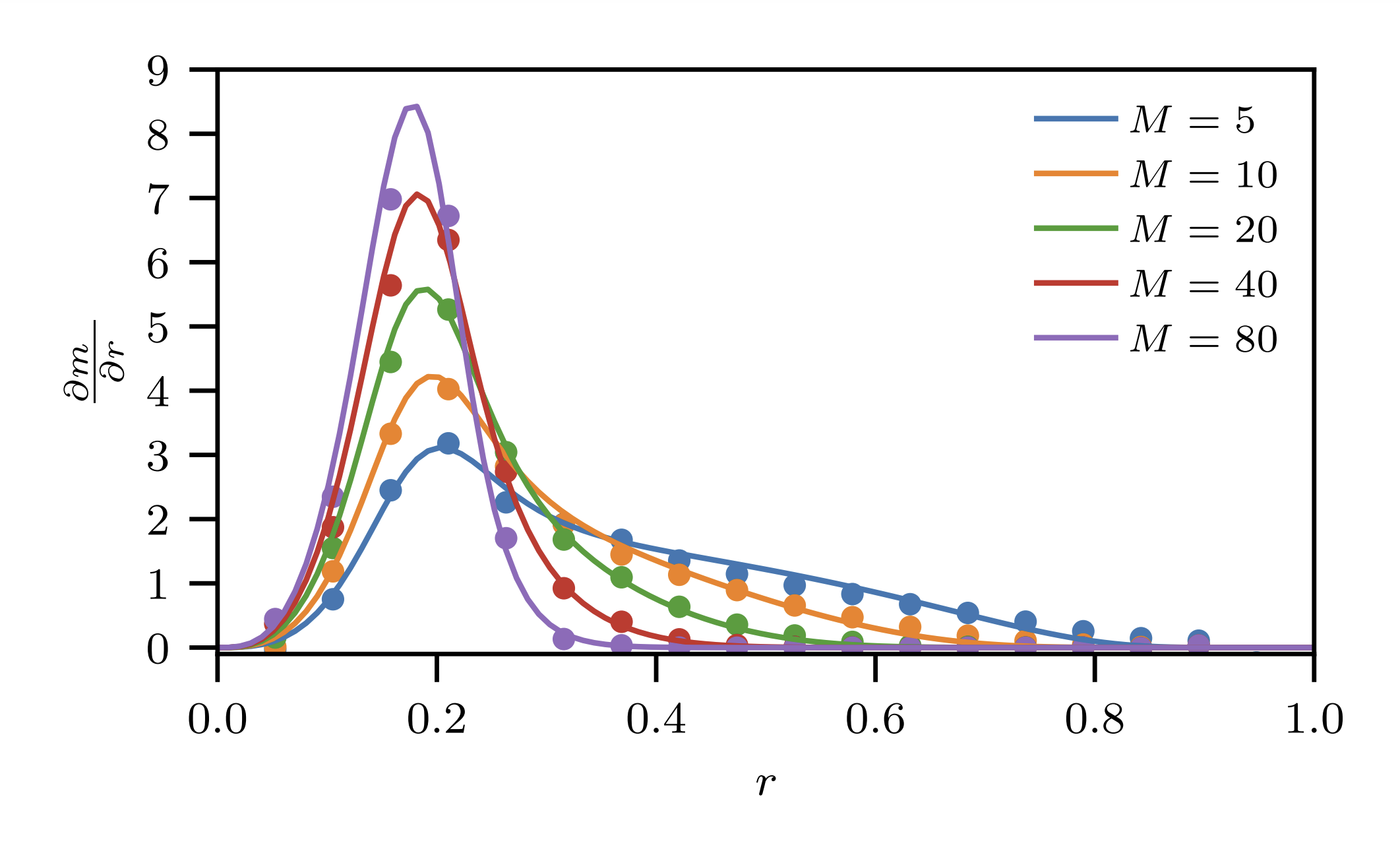}
    \caption{Comparisons between observables evaluated by MC simulations equipped with Plefka's dynamic (lines) and stability analysis (dots, see \eqref{eq:explicit}). 
    The number of examples $M$ varies as specified by the legend, while the number of neurons and patterns and are kept fixed at $N=6000, \ K=100$. As a consequence, the load is fixed below the critical value, $\gamma<\gamma_c$. 
    In particular, we report the  archetype magnetization $m$ and its susceptibility $\partial_{r}m$ at various training-set sizes $M$ by making the noise $r$ in the training set vary from $0$,  where all the example are pure random noise, to $1$, where there is no difference among examples and archetype.
    We note that in the small noise limit $r\to1$ the network always perfectly retrieves the archetype as expected, whereas, for $r\to 0$, no retrieval is possible.}
    \label{fig:s2n_unsup}
\end{figure}

\begin{figure}[t]
    \centering
    \includegraphics[scale=0.75]{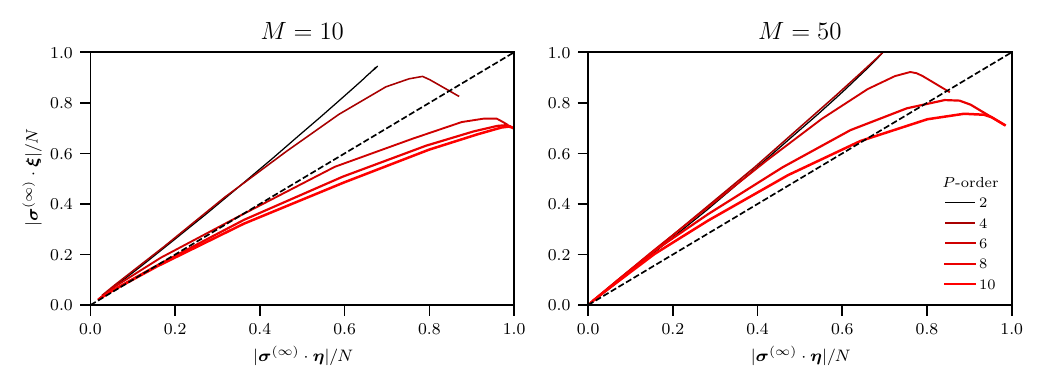}
\caption{Analysis of the capacity of the network to reconstruct an archetype generalizing from corrupted versions of it. The dataset is generated by $K=10$ Rademacher archetypes, each of size $N=784$, whence $M=10$ (left panel) and $M=50$ (right panel) examples are built for each archetype by setting $r=0.7$.  Then, we let the system relax from an initial configuration $\boldsymbol \sigma^{(0)}$ -- chosen as a corrupted version of one of the examples, say $\boldsymbol \eta^a$ -- to the thermalized configuration $\boldsymbol \sigma ^{(\infty)}$ by Plefka's dynamics. We determine the overlap between $\boldsymbol \sigma ^{(\infty)}$ and $\boldsymbol \eta^a$, as well as the overlap between $\boldsymbol \sigma ^{(\infty)}$ and the archetype $\boldsymbol \xi$. These quantities are then averaged over different initializations. Notice that these overlaps, i.e., the normalized scalar products, provide a measure of resemblance between the involved vectors; for instance, the Hamming distance between $\boldsymbol \xi$ and $\boldsymbol \sigma ^{(0)}$ is nothing but $\frac{1}{2} (N - \boldsymbol \xi \cdot \boldsymbol \sigma ^{(0)}$). 
The dashed line represents the identity and plays a reference: above this line the system has escaped from the attraction basin of the example $\boldsymbol \eta^a$ and has moved closer to the archetype. 
We notice that, as the interaction degree $P$ increases (see the legend), the attractivity of the archetype is impaired.
If we want that the curve remains above the threshold as $P$ increases, the number of examples has to be increased accordingly.  }\label{fig:comparison_unsup}
\end{figure}

In this section we present some results useful to check the effective performance of the network. First, we use the stability analysis to find an explicit expression for $\m$ in the noiseless limit also via this path. Next, we estimate the minimum number of examples, as a function of $P, r$, and $\gamma$, that we need for a successful retrieval. Finally, we show some outcomes obtained by MC simulations, concerning the reconstruction of the archetypes and critical load.

\label{sec:num}
\subsection{Stability analysis and Monte Carlo simulations} \label{sec:MC}
\label{subsec:S2N}

In this section we carry on a stability analysis in the noiseless limit: 
 we  suppose that the network is in a retrieval configuration, say $\bm \sigma = \bm \xi^1$ without loss of generality, we evaluate the local field $h_i(\bm \xi^1)$ acting on the generic neuron $\sigma_i$, and check that $h_i(\bm \xi^1)\si >0$ is satisfied for any $i=1, \hdots, N$; this condition ensures the stability of the retrieval configuration.

\par\medskip
We start by rearranging the cost function (\ref{def:H_PHopEx}) exploiting the mean-field nature of the model, namely 
\begin{equation}
    -\beta\mathcal{H}_{N,K,M,r}^{(P)}(\boldsymbol{\sigma} \vert \bm \eta)=\SOMMA{i=1}{N}h_i(\boldsymbol \sigma)\sigma_i
\end{equation}
where the local field $h_i(\boldsymbol{\sigma})$ acting on the $i$-th spin is 
\begin{align}
    h_i(\boldsymbol{\sigma})=& \dfrac{1}{ \R^{P/2}M N^{P-1}}\SOMMA{\mu=1}{K}\SOMMA{a=1}{M}\SOMMA{(i_2,\cdots,i_P)\neq i}{N,\cdots,N} \eta^{\mu, a}_{i_1}...\eta^{\mu,a}_{i_P}\sigma_{i_2}...\sigma_{i_P} .
\end{align}

Calling $O^{(n)}$ the $n$-th iteration of the MC Markov chain scheme regarding the generic observable $O$ and starting by a Cauchy condition where the neurons are aligned with the first pattern, i.e. $\boldsymbol \sigma^{(0)}=\boldsymbol \xi^1$, we update the neural configuration as
\begin{equation}
    \sigma^{(n+1)}_i=\sigma^{(n)}_i\mathrm{sign}\left[\tanh{\left(\sigma_i^{(n)}h_i^{(n)}(\boldsymbol{\sigma}^{(n)})\right)}+\Gamma_i\right] \;\;\mathrm{with}\;\;\Gamma_i\sim \mathcal{U}[-1;+1]
\end{equation}
and, performing the zero fast-noise limit $\beta \to\infty$, we have
\begin{equation}
    \sigma^{(n+1)}_i=\sigma^{(n)}_i\mathrm{sign}\Big[\sigma_i^{(n)}h_i^{(n)}\big(\boldsymbol{\sigma}^{(n)}\big)\Big].
\end{equation}
The one-step MC approximation for the magnetization is then
\begin{equation}\label{NoHinton}
    m_1^{(2)}:=\dfrac{1}{N}\SOMMA{i=1}{N}\xi_i^1\sigma_i^{(2)}=\dfrac{1}{N}\SOMMA{i=1}{N}\mathrm{sign}\left(\xi_i^{1}h_i^{(1)}(\boldsymbol{\xi}^{1})\right),
\end{equation}
and, in the thermodynamic limit $(N \to \infty)$ the argument of the sign function in the r.h.s. of Eq. (\ref{NoHinton}) can be approximated, by the CLT\footnote{Again, we have a sum of variables that are not Gaussian, but whose momenta are vanishing fast enough with $N$ to make the CLT appliable. However, unlike the case discussed in Sec.~\ref{sec:solution} , the overall sum is mathematically more treatable (because here the variables $z_{\mu,a}$ are missing) and we can estimate directly first and second moments.}, as $\xi_i^1 h_i^{(1)} \sim  \mu_1 + z_i \sqrt{\mu_2 - \mu_1^2}$, where $z_i \sim \mathcal{N}(0,1)$, and 
\begin{eqnarray}
&\mu_1& \coloneqq \mathbb{E}_{\xi}\mathbb{E}_{(\eta|\xi)}\left[\xi_i^{1}h_i^{(1)}(\boldsymbol{\xi}^{1})\right]\label{eq:s2n_mu1}
\\
&\mu_2& \coloneqq \mathbb{E}_{\xi}\mathbb{E}_{(\eta|\xi)}\left\{\left[ h_i^{(1)}(\boldsymbol{\xi}^{1})\right] ^2\right\}.\label{eq:s2n_mu2}
\end{eqnarray}
Then, recalling 
$$
\int_{-\infty}^{+\infty}   \frac{dz}{\sqrt{2 \pi}} e^{-\frac{z^2}{2}} \mathrm{sign} \left(\mu_1 + z \sqrt{\mu_2 - \mu_1^2}\right) = \mathrm{erf}\left(\dfrac{\mu_1}{\sqrt{2(\mu_2-\mu_1^2)}}\right), 
$$
for $N\gg 1$,  we get
\begin{equation}
    m_1^{(2)} \sim \mathrm{erf}\left(\dfrac{\mu_1}{\sqrt{2(\mu_2-\mu_1^2)}}\right).
    \label{eq:s2n_M}
\end{equation}

As reported in Appendix \ref{app:momenta},
first and second momenta of $\xi_i^1 h_i^{(1)}(\bm \xi^1)$ read as
\begin{eqnarray}
\mu_1&=& \dfrac{1}{(1+\rho)^{P/2}}
\\
\mu_2&=& \left(\dfrac{1}{(1+\rho)^{P/2}}\right)^2\Big[\alpha_{P-1}(P-1)!(1+\rho_{_P})+1+\rho\Big]
\end{eqnarray}
where, we introduced $\rho_{_P}:=\frac{1-r^{2P}}{M r^{2P}}$ as a generalization of the dataset entropy $\rho=\frac{1-r^2}{M r^2}$ .

Using the $P$-independent load (see Eq. \eqref{eq:alphaPP-1}), we can write 
\begin{equation}
    \mu_2-\mu_1^2=\left(\dfrac{1}{(1+\rho)^{P/2}}\right)^2\left[\dfrac{2}{P}\gamma\left(1+\rho_{_P}\right)+\rho\right],
\end{equation}
and we obtain the following explicit expression for the one-step MC magnetization
\begin{equation} \label{eq:explicit}
        m_1^{(2)}\sim \mathrm{erf}
        \Big \{ \Big [  \dfrac{4\gamma}{P}(1+\rho_{_P})+ 2 \rho \Big]^{-\frac{1}{2}}\Big \}.
\end{equation}
Thus, we can recast the condition determining if the network successfully retrieves one of the archetypes by requiring that this one-step MC magnetization is larger than  $\mathrm{erf}(\Theta)$ where $\Theta\in\mathbb{R}^+$ is a tolerance level, thus we obtain 
\begin{equation}
    \dfrac{1}{\sqrt{2\Big[\gamma\dfrac{2}{P}(1+\rho_{_P})+\rho\Big]}}>\Theta.
    \label{eq_S2N_stability_condition}
\end{equation}
Otherwise stated, in order to retrieve (under a confidence level $\Theta$) a given archetype starting from a perturbed versions of it, one has to fulfil the following condition 
\begin{equation}
    1>2\Theta^2\left[\rho+\gamma \dfrac{2}{P}(1+\rho_{_P})\;\right].
    \label{eq:M_RS_aP-1}
\end{equation}
 Setting the confidence level $\Theta=1/\sqrt{2}$, which corresponds to the condition 
 \begin{align}
 \mathbb{E}_{\xi}\mathbb{E}_{(\eta|\xi)}[\xi_i^1h_i^{(1)}(\boldsymbol \xi^1)]>\sqrt{\mathrm{Var}[\xi_i^1h_i^{(1)}(\boldsymbol \xi^1)]}
 \end{align} 
 (namely we have a non-null magnetization in \eqref{eq:s2n_M}), the previous relation determines a lower bound for $M$ that is denoted as $M_\otimes(r,P,\gamma)$.
%
%
%
%
%
\subsection{Critical load and bounds for the dataset size}

\begin{figure}
    \centering
    \includegraphics[width=7.5cm]{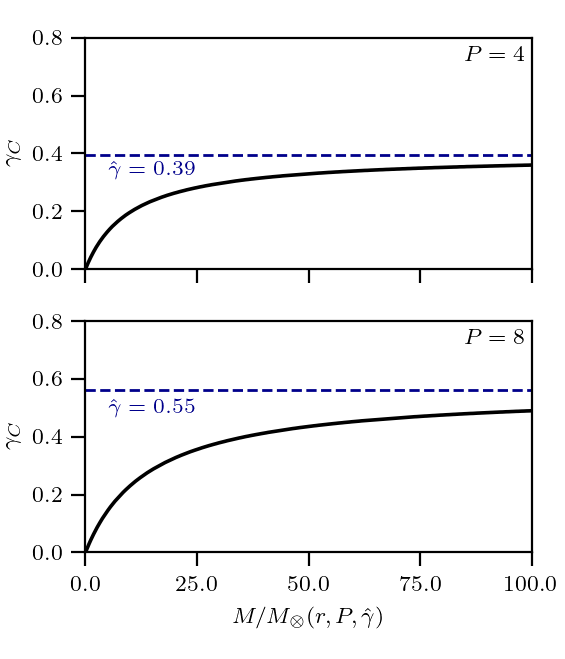}
    \includegraphics[width=7.5cm]{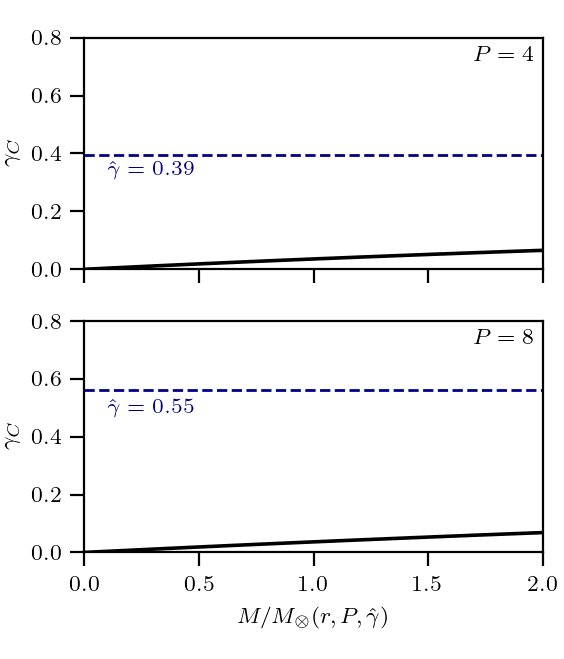}
    \caption{We numerically solve the self equations in the $T\to 0$ limit, namely \eqref{GroundZero}, for $r=0.2$. In these four panels we plot the critical load $\gamma_c$ vs the order of magnitude number of examples $M$ w.r.t. $M_\otimes (r, P, \hat{\gamma})$, where $\hat\gamma$ is the critical load of the standard dense Hebbian network with the same interaction order, for different degrees of interaction $P=4,8$, at work with the simpler storing protocol. We notice that $\gamma_c$ increases with $M$ and, for $M \gg M_{\otimes}(r, P, \gamma=\hat\gamma)$, it saturates to $\hat\gamma$ \cite{EmergencySN}, that is represented by the horizontal dashed line.  
    }
    \label{fig:my_label5}
\end{figure}

\begin{figure}[t]
    \centering
    \includegraphics[width = 15cm]{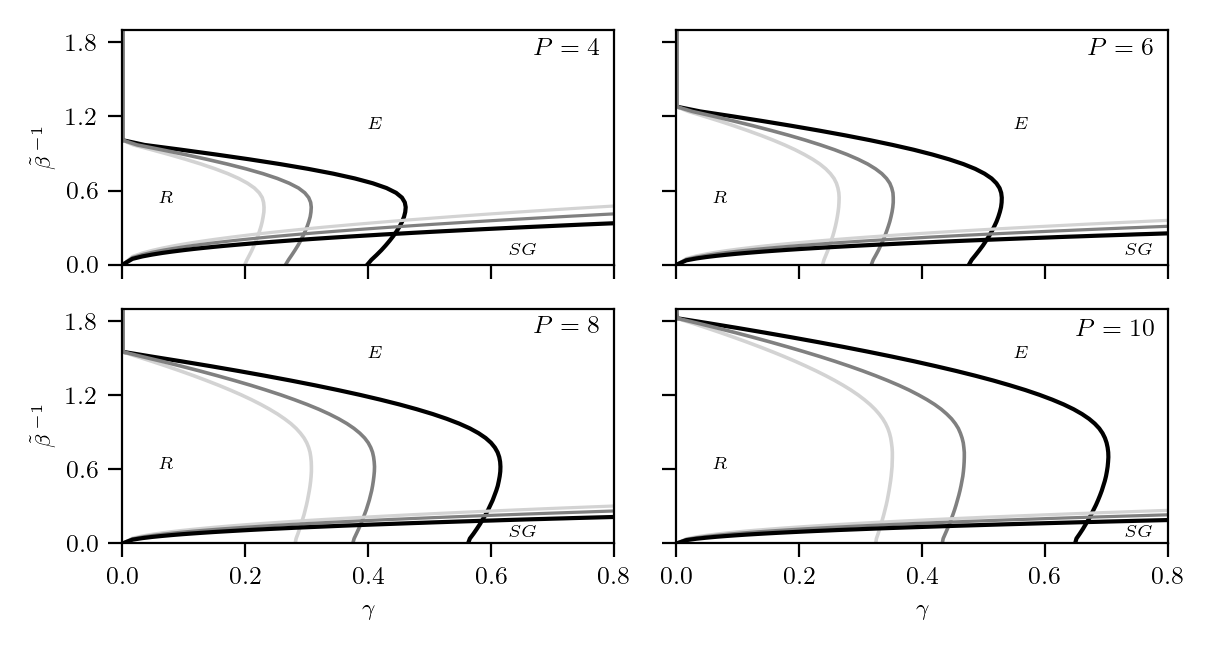}
    \caption{Each panel represents the phase diagram of the dense Hebbian networks trained -with no supervision- at different values of $P$, as specified. In any case we set $r=0.2$ and we compare outcomes for   $M=M_\otimes(r,P,\gamma)$ (light gray) and $M=2M_\otimes(r,P,\gamma)$ (dark gray), where
    $M_\otimes(r,P,\gamma)$ is given by the expression in equation \eqref{eq:loadM}. We notice that, as $P$ increases, the transition lines approach those of the dense Hebbian storage limit (black solid line). Moreover, the instability region, caused by the overlap between retrieval and spin-glass regions, decreases as $P$ increases. Note that the recess of the maximal storage as the temperature goes to zero is a signature of replica symmetry breaking, that is not addressed here (see \cite{Albanese2021}).}
    \label{fig:my_label6}
\end{figure}

We now discuss a few special cases for $M_\otimes (r, P , \gamma)$ under the assumption $r \ll 1$. 
In the low-load regime $\gamma=0$, the expression in \eqref{eq:M_RS_aP-1} becomes 
\begin{equation}
    M>\left(\dfrac{1}{\sqrt{2}}\right)^{\2}\dfrac{(1-r^{^2})}{r^{\2}}\sim \dfrac{1}{2}\dfrac{1}{r^2}\Longrightarrow M_\otimes(r,P,0)= \dfrac{1}{2r^{\2}}
\end{equation}
where the last equality holds for $r \ll 1$.
\newline
If $\gamma \neq 0$ and $P=2$, i.e. Hopfield classic case, as shown in \cite{AgliariDeMarzo}
\begin{equation}
    M>\Theta^{\2}\left(\dfrac{(1-r^{^2})}{r^{\2}}+\gamma\dfrac{1}{r^{\4}}\;\right)\sim \dfrac{1}{2} \gamma \dfrac{1}{r^4}\Longrightarrow M_\otimes(r,2,\gamma)= \gamma\dfrac{1}{2r^{\4}}.
\end{equation}
Finally, if $\gamma \neq 0$ and $P>2$, we have
\begin{equation}
\small
\begin{array}{lll}
\label{eq:loadM}
    M>&\left(\dfrac{1}{\sqrt{2}}\right)^{\2}\left(\dfrac{(1-r^{^2})}{r^{\2}}+\gamma \dfrac{2}{P}\dfrac{1}{r^{2P}}\;\right)\sim \dfrac{1}{2}\gamma\dfrac{2}{P} \dfrac{1}{r^{2P}} \Longrightarrow M_\otimes(r,P,\gamma)=  \gamma\dfrac{1}{P}\dfrac{1}{r^{2P}}.
    \end{array}
\end{equation}
This result implies that we need a larger number of examples if we use dense networks w.r.t. Hopfield pairwise networks, further, we stress that -whatever the case- we always end up with power laws thresholds for learning relating the critical amount of examples to the dataset noise.

We notice that if in \eqref{eq:g_of_unsuper_n} we set either $r \to 1$ and $M\to 1$ (i.e., we give the original patterns and not the examples, see Eq. \eqref{eq:Bernoulli}, and, so, we have $n_{1,a}=m_1$ in Def. \eqref{def:orderparam}), or $M\gg M_\otimes(r,P,\gamma)$ (i.e., we give the network a very large number of examples), we recover dense Hebbian neural network at work with the simpler storing protocol.


\par\medskip
The results obtained analytically in this section are corroborated by numerical simulations and the related outputs are collected in Figs.~\ref{fig:s2n_unsup}-\ref{fig:my_label6}. We stress that, to avoid the computational-expensive updating of the synaptic tensor in \eqref{def:H_PHopEx}, we implemented a Plefka's dynamics in our MC scheme. This is an effective dynamics that allows us to keep track of the evolution of the network's order parameters at the level of their mean values: we refer to Appendix \ref{app:plefka} for more details.
Let us now comment these numerical results. \\
A corroboration of the goodness of Plefka's dynamics is given in Fig.~\ref{fig:s2n_unsup}, where we show a comparison between MC simulation with Plefka's dynamics and stability analysis. \\ 
Figure \ref{fig:comparison_unsup} shows evidence that, for a given choice of $N, K, M$, and $r$, when the degree of interaction $P$ increases, the network is no longer able to generalise the archetype from the examples. This suggests that, if we want to retain the reconstruction capabilities, the number of examples $M$ should scale with $P$, as showed in \eqref{eq:loadM}.
\newline
Tying in with this speech, in Fig. \ref{fig:my_label5} we plot the number of examples $M$ w.r.t. the critical load $\gamma_c$, for different values of $P$. We recall that the critical load $\gamma_c$  is the load beyond which a black-out scenario emerges, namely $\lim\limits_{\gamma\to\gamma_c^{-}}\m\neq 0$ and $\lim\limits_{\gamma\to \gamma_c^{+}} \m= 0$. The black dotted line represents the critical load of the Hopfield dense neural network as a reference. We can see that, when $M$ is chosen following the prescription in \eqref{eq:loadM}, we can reach the performances of Hopfield dense network. \\
Finally, in Fig. \ref{fig:my_label6} we show the phase diagrams in the space $(\tilde \beta, \gamma)$ for different values of $P$ (each corresponding to a different panel). Interestingly, as $M$ increases, the retrieval zone gets wider.

\section{Conclusion and outlooks} \label{sec:conclusions}

In this paper we investigated the information processing capabilities of dense Hebbian networks  endowed with couplings stemmed by an unsupervised protocol for learning. In order to have a mathematically  tractable theory, this is developed for random structureless datasets. 
The network is made of $N$ neurons that interact in groups of $P$ units with a strength encoded by the synaptic tensor $\boldsymbol J^{(unsup)}$, whose generic entry (retaining only the leading order) reads as
\begin{equation}
J^{(unsup)}_{i_1i_2...i_P} \sim \frac{1}{ M N^{P-1}} \sum_{\mu=1}^K \sum_{a=1}^M \eta_{i_1}^{\mu,a} \eta_{i_2}^{\mu,a}...\eta_{i_P}^{\mu,a},
\end{equation}
where $\{\boldsymbol \eta^{\mu} \}_{a=1,...,M}^{\mu=1,...,K}$ is the dataset available, made of $K$ subsets of examples (labeled by $a$) referred to $K$ unknown archetypes. The quality of the dataset, namely how ``far'' these examples are, in the average, from the related archetype, is ruled by $r$ (such that by setting $M=1$ and $r=1$ we recover the standard dense Hebbian network under the simpler storage prescription\cite{Baldi, Bovier,Albanese2021}).

Hereafter we summarize the main outcomes of our work. 
As far as general neural network's theory is concerned
\begin{enumerate}
\item The dense  Hebbian network under the simpler storage prescription is well-known to be able to store a number of patterns that grows as $K \sim N^{P-1}$. This {\em high-load regime}, is preserved when the standard Hebbian coupling is replaced by the unsupervised Hebbian coupling. One can still introduce a load $\alpha_{P-1} = \lim_{N \to \infty} \frac{K}{N^{P-1}}$ and determine a critical value beyond which a black-out scenario emerges: notably, this value does not depend on the dataset properties and it is solely a network's characteristic.   

\item For a correct learning -and subsequent retrieval- of the archetype, there exists a threshold value $M_{\otimes}$ to overcome and this scales as $M_{\otimes} \propto 1/(P \, r^{2P})$. Thus, when using these dense machines in the unsupervised regime, one can actually reconstruct up to $K \sim N^{P-1}$ archetypes only under the condition of a suitably large number of required examples available. In other words, increasing the number of retrievable patterns by a factor $N$ (which means increasing the interaction order of one unit) requires an amplification in the number of examples per archetype of about $1/r^2$.
Increasing the dataset size beyond $M_{\otimes}$ leads to a wider retrieval region, namely to a larger critical load and to a larger critical temperature. The large cost in terms of available data is a peculiarity of the unsupervised regime,  in fact, in supervised dense networks, $M_{\otimes}$ does not scale with $P$, see \cite{super}. 

\item There is another intriguing feature displayed by dense networks that is preserved in the unsupervised regime. Indeed, the
reconstruction is feasible also in the presence of an extensive noise affecting its coupling and yielding to a signal-to-noise ratio increasing algebraically with $N$. Again, to mitigate the effects of this noise one has to move to a low-load regime in order to generate redundancy in the information allocated in the coupling tensor.

\end{enumerate}
As far as the computational and mathematical technicalities are concerned
\begin{enumerate}
\item in this dense scenario, the post-synaptic potential does not have a Gaussian shape and standard techniques (e.g. replica trick, interpolation approaches) do not work straightforwardly, however, it is possible to adapt them by applying the CLT and restoring an effective Gaussian framework whose validity is corroborated by numerical simulations.
\item as $P$ grows, the synaptic tensors become very expensive to evaluate (and prohibitive to  be updated during learning dynamics): to overcome this problem, we adapted the Plefka's approximation to the case, resulting in a remarkable speed up of the simulations, yet preserving an extremely good accuracy in the results.
\end{enumerate}
Overall these technical extensions are of broad generality and can be applied to several other neural networks.



\appendix

\section{Proof of Proposition \ref{P_quenched}} 
\label{app:proofP-1}

In order to prove Proposition \ref{P_quenched},  
we need the following
\begin{lemma} 
The $t$ derivative of the interpolating quenched pressure \eqref{hop_GuerraAction} is given by 
\begin{equation}
    \begin{array}{lll}
         \dfrac{d \mathcal{A}^{(P)}_{N,K,M, r, \beta}(t)}{d t}\coloneqq&  \dfrac{\beta '}{2M}(1+\rho)^{P/2}\SOMMA{a=1}{M}\left(\l n_{1,a}^{^P}\r_t-\dfrac{2 M\,\psi}{\beta '(1+\rho)^{P/2}}\l n_{1,a}\r_t \right)
         \\\\
        & -\dfrac{A^{\2}}{2}\Big(1-\l q_{12} \r_t\Big) +\dfrac{\b\,^2(1+\rho_P)}{4 (1+\rho)^{P}}\dfrac{K P!}{2N^{P-1}}  \Big(1 -\l q_{12}^{^{P}}\r_t\Big).
    \end{array}
    \label{eq:streaming_RS_Guerra}
\end{equation}
\normalsize
where we use $\rho_P=\frac{1-r^{2P}}{Mr^{2P}}$.
\end{lemma}

\begin{proof}
Deriving Eq. \eqref{hop_GuerraAction} with respect to $t$, we get 
\begin{equation}
\begin{array}{lll}
     \dfrac{d \mathcal{A}^{(P)}_{N,K,M, r,\beta}(t)}{d t}&=& \dfrac{1}{N} \mathbb{E} \dfrac{1}{\mathcal{Z}_{N, K,M, r,\beta }^{(P)}} \SOMMA{\bm \sigma}{} \mathcal B^{(P)}_{N, K,M, r,\beta}  \left[ \frac{\beta^{'}N}{2M}(1+\rho)^{P/2}\sum_{a=1}^M n_{1,a}^P - \psi N\sum_{a=1}^M n_{1,a} \right.
     \\\\
     &&\left.  - \dfrac{1}{2\sqrt{1-t}}A\SOMMA{i=1}{N}Y_i\sigma_1 +\dfrac{1}{2} \dfrac{\beta^{'} P!\sqrt{(1+\rho_P)}}{2t\,(1+\rho)^{P/2}N^{P}}\SOMMA{\mu\geq 2}{K}\SOMMA{(i_{_1},\cdots, i_{_{P}})}{N,\cdots,N}\lambda^{\mu}_{i_1\cdots i_P}\sigma_{_{i_1}}\cdots\sigma_{_{i_{_{P}}}}\right]
     \\\\
     &&= \dfrac{\beta^{'}}{2M}(1+\rho)^{P/2}\SOMMA{a=1}{M} \langle n_{1,a}^P \rangle_t - \psi \SOMMA{a=1}{M} \langle n_{1,a} \rangle_t +D_1+D_2.
\end{array}
\label{eq:proof_streaming_Guerra}
\end{equation}
Now, using Stein's lemma\footnote{This lemma, also known as Wick's theorem, applies to standard Gaussian variables, say $J \sim \mathcal N(0, 1)$, and states that, for a generic
function $f(J)$ for which the two expectations $\mathbb{E}\left( J f(J)\right)$ and $\mathbb{E}\left( \partial_J f(J)\right)$ both exist, then
\begin{align}
    \label{eqn:gaussianrelation2}
    \mathbb{E} \left( J f(J)\right)= \mathbb{E} \left( \frac{\partial f(J)}{\partial J}\right).
\end{align}
} on the random variables $Y_i$ and $\lambda^{\mu}_{i_1 \hdots i_{P}}$,
we may rewrite the last two terms of \eqref{eq:proof_streaming_Guerra} as

\begin{equation}
\begin{array}{lll}
     D_1&= 
    -\displaystyle{\frac{1}{2N\sqrt{1-t}}} B\SOMMA{i=1}{N} \mathbb{E} \partial_{Y_i} \left[ \frac{1}{\mathcal{Z}_{N, K,M, r, \beta }^{(P)}} \sum_{\bm \sigma} \mathcal B^{(P)}_{N, K,M, r,\beta}  \sigma_i \right]&=-\dfrac{A^{\2}}{2}\Big(1-\l q_{12} \r_t\Big),
\end{array}
\label{eq:D2_GuerraRS}
\end{equation}
\begin{equation}
\begin{array}{lll}
     D_2&=  \displaystyle{ \frac{\beta^{'}\sqrt{(1+\rho_P)}}{4t\,(1+\rho)^{P/2}  N^{P+1}}} \SOMMA{\mu\geq 2}{K}\SOMMA{(i_{_1},\cdots, i_{_{P}})}{N,\cdots,N} \mathbb{E} \partial_{\lambda^{\mu}_{i_1,\cdots,i_P}} \left[ \frac{1}{\mathcal{Z}_{N, K,M, r, \beta }^{(P)}} \sum_{\bm \sigma}\mathcal B^{(P)}_{N, K,M, r,\beta}  \sigma_{_{i_1}}\cdots\sigma_{_{i_{_{P}}}} \right]
    \\\\
    \label{eq:D3_GuerraRS}
    &=\dfrac{\b\,^2 (1+\rho_P)}{4 (1+\rho)^{P}}\dfrac{KP!}{2N^{P-1}}  \Big(1 -\l q_{12}^{^{P}}\r_t\Big).
\end{array}
\end{equation}
\normalsize
Rearranging together \eqref{eq:D2_GuerraRS} and \eqref{eq:D3_GuerraRS} we obtain the thesis.
\end{proof}

\begin{assumption}
\label{ass: RSassumption}
	As a consequence of the RS assumption, for the generic order parameter $X$, being $\Delta X \coloneqq  X - \bar{X}$, the deviation w.r.t. the expectation value, then $$\langle (\Delta X)^2 \rangle_t \xrightarrow[]{N\to\infty}0$$
	and, clearly, the RS approximation also implies that, in the thermodynamic limit, $\langle \Delta X \Delta Y \rangle_t \to 0$ for any generic pair of order parameters $X,Y$. Moreover in the thermodynamic limit, we have $\langle (\Delta X)^k \rangle_t \rightarrow 0$ for $k \geq 2$.
\end{assumption}
Hereafter, in order to lighten the notation, we will drop the subscript $t$.
%
In the following we can use the relation
\begin{equation}
 \begin{array}{lll}
     \langle x^P \rangle -P\,\bar{x}^{P-1}\langle x \rangle &=& -(P-1) \bar{x}^{P}+ \SOMMA{k=2}{P} \begin{pmatrix}P\\k\end{pmatrix} \langle (x-\bar{x})^k \rangle \bar{x}^{P-k} ,
\end{array}  
\label{eq:RS_pq_Potenziali}
\end{equation}
for any order parameter $x$ with equilibrium value $\bar{x}$, which is computed straightforwardly by Newton's binomial \cite{GuysAlone}. 

Using these relations, if we fix the constants $\psi, A$, appearing in the interpolating partition function introduced in Definition \ref{def:part_Interpolante_RS}, as
\begin{equation}
    \begin{array}{lll}
         &\psi=\beta '\dfrac{P}{2M}(1+\rho)^{P/2}\n^{P-1} , 
         \
         &A^2=\dfrac{\b \,^2(1+\rho_P)}{(1+\rho)^{P}} \dfrac{P}{2} \dfrac{K P!}{2 N^{P-1}}\q^{^{P-1}}\,,
    \end{array}
\end{equation}
we can rewrite the derivative of the interpolating pressure w.r.t. $t$ as 
\begin{equation}
\begin{array}{lll}
     \dfrac{d \mathcal{A}_{N,K,M, r, \beta}^{(P)}(t)}{d t}&\coloneqq&  -\dfrac{\beta '}{2}(P-1)(1+\rho)^{P/2}\n^P+\dfrac{\b\,^2(1+\rho_P)}{4 (1+\rho)^{P}}\dfrac{K P!}{2N^{P-1}}(1-P\q^{P-1}+(P-1)\q^P) \\\\
     &+& \dfrac{\b}{2M}(1+\rho)^{P/2}\SOMMA{a=1}{M} \SOMMA{k=2}{P} \begin{pmatrix}P\\k\end{pmatrix} \langle (n_{1,a}-\bar{n})^k \rangle \bar{n}^{P-k}\\\\
     &+&\dfrac{\b\,^2(1+\rho_P)}{4 (1+\rho)^{P}}\dfrac{K P!}{2N^{P-1}}\SOMMA{k=2}{P} \begin{pmatrix}P\\k\end{pmatrix} \langle (q_{12}-\bar{q})^k \rangle \bar{q}^{P-k}.
\end{array}
\label{eq:streaming_RS_Guerra2}
\end{equation}
\normalsize



\begin{proof}

\textit{(of Proposition \ref{P_quenched})} Exploiting the fundamental theorem of calculus, we can relate $\mathcal{A}_{N, K,M, \beta , r}^{(P)}(t=1)$ and $\mathcal{A}_{N, K,M, \beta , r}^{(P)}(t=0)$ as

\begin{equation}
    \mathcal{A}_{N, K,M, \beta , r}^{(P)}=\mathcal{A}_{N, K,M, r, \beta}^{(P)}(t=1)=\mathcal{A}_{N, K,M, r, \beta}^{(P)}(t=0)+\int\limits_0^1\, \partial_s \mathcal{A}_{N, K,M, r, \beta}^{(P)}(s
)\Big|_{s=t}\,dt.
\label{eq:F_T_Calculus}
\end{equation}
\normalsize
We have just computed the derivative w.r.t. $t$, which is \eqref{eq:streaming_RS_Guerra2}; all we need is to recover the one-body term:
\begin{equation}
\begin{array}{lll}
     \mathcal{A}_{N, K,M, r,\beta}^{(P)}(t=0)=\mathbb{E}&\left\{ \ln 2\cosh\left[J\xi^1+\beta '\dfrac{P}{2}\n^{P-1} (1+\rho)^{P/2-1}\resub\etaM \right. \right. \\\\
     &\left. \left.+Y \sqrt{\dfrac{\b \,^2(1+\rho_P)}{(1+\rho)^{P}} \dfrac{P}{2} \dfrac{K P!}{2 N^{P-1}}\q^{^{P-1}}}\right]\right\}.
    \end{array}
    \label{eq:one_body_GuerraRS_Nfinito}
\end{equation}
Putting \eqref{eq:streaming_RS_Guerra2} and \eqref{eq:one_body_GuerraRS_Nfinito} in \eqref{eq:F_T_Calculus}, we find 
\begin{equation}
\label{eq:pressure_GuerraRS_Nfinito}
\begin{array}{lll}
     &\mathcal{A}_{N,K,M, r, \beta}^{(P)} = \mathbb{E}\left\{ \ln{2\cosh{\left[J\xi^1+\beta '\dfrac{P}{2}\n^{P-1} (1+\rho)^{P/2-1}\resub\etaM+Y \sqrt{\dfrac{\b \,^2(1+\rho_P)}{(1+\rho)^{P}} \dfrac{P}{2} \dfrac{K P!}{2 N^{P-1}}\q^{^{P-1}}}\right]}}\right\}
         \\\\
         &-\dfrac{\beta '}{2}(P-1)(1+\rho)^{P/2}\n^P+\dfrac{\b\,^2(1+\rho_P)}{4 (1+\rho)^{P}}\dfrac{K P!}{2N^{P-1}}(1-P\q^{P-1}+(P-1)\q^P)+ \displaystyle\int\limits_0^1 V_{N,M}(t)dt.
\end{array}
\end{equation}
where $\mathbb{E}=\mathbb{E}_\xi\mathbb{E}_{(\eta|\xi)}\mathbb{E}_Y$, $\etaM=\frac{1}{rM}\SOMMA{a=1}{M}\eta^{1,a}$ and the potential $V_{N,M}(t)$ is
\begin{equation}
\begin{array}{lll}
    V_{N,M}(t) =&    \dfrac{\b(1+\rho)^{P/2}}{2M}\SOMMA{a, k=1}{M, P}  \begin{pmatrix}P\\k\end{pmatrix} \langle (n_{1,a}-\bar{n})^k \rangle \bar{n}^{P-k}+\dfrac{\b\,^2(1+\rho_P)K P!}{8 (1+\rho)^{P}N^{P-1}}\SOMMA{k=2}{P} \begin{pmatrix}P\\k\end{pmatrix} \langle (q_{12}-\bar{q})^k \rangle \bar{q}^{P-k}
\end{array}
\end{equation}

Now, we know that for  $b=P-1$ we have $K =\alpha_{P-1} N^{P-1} + O(N^{P-1-\epsilon})$, $\epsilon >0$; thus, neglecting the lower terms, we have
\begin{equation}
\label{eq:pressure_GuerraRS_NN}
\begin{array}{lll}
     \mathcal{A}_{\alpha_{P-1},M, r, \beta}^{(P)} =& \mathbb{E}\left\{ \ln{2\cosh{\left[J\xi^1+\beta '\dfrac{P}{2}\n^{P-1} (1+\rho)^{P/2-1}\resub\etaM+Y \sqrt{\alpha_{P-1}\dfrac{P!}{2}\dfrac{\b \,^2(1+\rho_P)}{(1+\rho)^{P}} \dfrac{P}{2} \q^{^{P-1}}}\right]}}\right\}
         \\\\
         &-\dfrac{\beta '}{2}(P-1)(1+\rho)^{P/2}\n^P+\alpha_{P-1}\dfrac{P!}{2}\dfrac{\b\,^2(1+\rho_P)}{4 (1+\rho)^{P}}(1-P\q^{P-1}+(P-1)\q^P).
\end{array}
\end{equation}

Finally, we maximise the statistical pressure in \eqref{eq:pressure_GuerraRS_NN} w.r.t. the order parameters and we find 
	\begin{equation}
    \begin{array}{lll}
         \n=\dfrac{1}{1+\rho}\mathbb{E}\left\{\tanh{\left[\beta '\dfrac{P}{2}\n^{P-1} (1+\rho)^{P/2-1}\resub\etaM+Y \sqrt{\alpha_{P-1}\dfrac{P!}{2}\dfrac{\b \,^2(1+\rho_P)}{(1+\rho)^{P}} \dfrac{P}{2} \q^{^{P-1}}}\right]\etaM}\right\},
         \\\\
         \q=\mathbb{E}\left\{\tanh{}^{\2}{\left[\beta '\dfrac{P}{2}\n^{P-1} (1+\rho)^{P/2-1}\resub\etaM+Y \sqrt{\alpha_{P-1}\dfrac{P!}{2}\dfrac{\b \,^2(1+\rho_P)}{(1+\rho)^{P}} \dfrac{P}{2} \q^{^{P-1}}}\right]}\right\}.
         \label{eq:n_selfRS}
    \end{array}
\end{equation}
Putting the definition of P-independent load from \eqref{eq:alphaPP-1} in \eqref{eq:n_selfRS} we reach the thesis. 
\end{proof}


\begin{corollary}
In the large dataset limit, in the high-storage regime, the RS self-consistency equations can be expressed as
\begin{align}
    \n&= \dfrac{\m}{1+\rho} +\beta ' \dfrac{P}{2}\rho(1+\rho)^{P/2-2}(1-\q)\n^{^{P-1}}.\notag
    \\
    \q&=\mathbb{E}_{_Z}\left[\tanh{}^{\2}{ g(\beta, Z, \bar{n})}\right].\label{eq:n_of_M_unsup_app}
    \\
    \m&= \mathbb{E}_{_Z}\left[\tanh{ g(\beta, Z, \bar{n})}\right].\notag
\end{align} 
where 
 \begin{align}
    &g(\beta, Z, \bar{n})=\beta '\dfrac{P}{2}\n^{^{P-1}}(1+\rho)^{P/2-1}+\beta 'Z\sqrt{\rho\dfrac{P^2}{4}\n^{^{2P-2}}(1+\rho)^{P-2}+ \alpha_{P-1}\dfrac{P!}{2}\dfrac{(1+\rho_P)}{(1+\rho)^{P}} \dfrac{P}{2} \q^{^{P-1}}\;}\;.
    \label{eq:g_of_unsuper_n_app}
 \end{align}
\end{corollary}
\begin{proof}
In large dataset limit we can use the CLT so that we have 
\begin{equation}
\begin{array}{lll}
      \mathbb{E}_{(\eta|\xi)}[\etaM] = \xi^1\,,&\;\;\;\;\;&  \mathbb{E}_{(\eta|\xi)}[(\etaM)^2]-\Big(\mathbb{E}_{(\eta|\xi)}[\etaM]\Big)^2 = \rho (\xi^1)^2
\end{array}
\end{equation}
thus we get
\begin{equation}
    \etaM\sim \xi^1\left(1+\lambda\sqrt{\rho}\right).
    \label{eq:CLT_unsup}
\end{equation}
where $\lambda$ is a standard Gaussian variable $\lambda\sim \mathcal{N}(0,1)$. Now, replacing \eqref{eq:CLT_unsup} in the self-consistency equation for $\bar{n}$ in \eqref{eq:High_store_self_n_q}, applying Stein's lemma and exploiting the self-consistency equations for $\m$ and $\q$ in \eqref{eq:High_store_self_n_q} we get \eqref{eq:n_of_M_unsup_app}. Replacing this new expression of $\n$ in  the argument of the hyperbolic tangent of \eqref{eq:High_store_self_n_q}  and exploiting the parity of the hyperbolic tangent, we can explicitly compute the mean over $\xi$ 

\footnotesize

\begin{align}
    &\n=\dfrac{1}{1+\rho}\mathbb{E}_{\xi, \lambda,Y}\left\{\tanh\left[\beta '\left(\dfrac{P}{2}\n^{^{P-1}}(1+\rho)^{P/2-1}\left(1+\lambda\sqrt{\rho}\right)\xi^1+Y\sqrt{  \dfrac{\alpha_{P-1}P!(1+\rho_P)}{2(1+\rho)^{P}} \dfrac{P}{2} \q^{^{P-1}}\;}\right)\right]\left(1+\lambda\sqrt{\rho}\right)\xi^1\right\}\;\notag
    \\
    &=\dfrac{1}{1+\rho}\mathbb{E}_{\lambda,Y}\left\{\tanh\left[\beta '\left(\dfrac{P}{2}\n^{^{P-1}}(1+\rho)^{P/2-1}\left(1+\lambda\sqrt{\rho}\right)+Y\sqrt{  \alpha_{P-1}\dfrac{P!}{2}\dfrac{(1+\rho_P)}{(1+\rho)^{P}} \dfrac{P}{2} \q^{^{P-1}}\;}\right)\right]\left(1+\lambda\sqrt{\rho}\right)\right\}
 \end{align} 
\normalsize
Now we use the relation
\begin{align}
    \mathbb{E}_{\lambda,Y}[F(a_1+ \lambda a_2+Y a_3)]=\mathbb{E}_{_Z}\left[F\left(a_1+Z\sqrt{a_2^2+a_3^2}\right)\right],
\end{align}
with $\lambda$, $Y$ and $Z$ i.i.d. Gaussian random variables, $\lambda$, $Y$ $Z \sim \mathcal{N}(0,1)$ and we have put $F(a_1+\lambda a_2+Y a_3)= \tanh(a_1+\lambda a_2+Y a_3)$, $a_1=\beta '\dfrac{P}{2}\n^{^{P-1}}(1+\rho)^{P/2-1}$, $a_2=a_1\sqrt{\rho}$, $a_3=\beta '\sqrt{\alpha_{P-1}\frac{P!}{2}\dfrac{\b \,^2(1+\rho_P)}{(1+\rho)^{P}} \frac{P}{2} \q^{^{P-1}}\;}$\\ In this way we can reduce the number of Gaussian averages to a single one and reach the thesis.
\end{proof}

\begin{corollary}
The self-consistency equations of the dense Hebbian neural networks in the unsupervised setting in the high-storage regime  and in null-temperature limit $\beta \to \infty$ are
\begin{equation}
    \begin{array}{lll}
         \m = \mathrm{erf}\left[\dfrac{P}{2}\dfrac{\m^{P-1}}{G}\right],\;\;\;\;\q=1,
          \\\\
         G=\sqrt{2\left[\rho\left(\dfrac{P}{2}\m^{^{P-1}}\right)^2+  \alpha_{P-1}\dfrac{P!}{2}(1+\rho_P) \dfrac{P}{2} \right]\;}.
    \end{array}
\end{equation}
\end{corollary}

\begin{proof}
For this proof it is convenient to introduce an additional term $\tilde \beta x$ in the argument of the hyperbolic tangent ($g(\Tilde{\beta},Z,\m)$) in  \eqref{eq:g_of_unsuper_n_app}
\begin{equation}
    \begin{array}{lll}
        \q&=&\mathbb{E}_{_Z}\left[\tanh{}^{\2}\left({ \tilde\beta \dfrac{P}{2}\m^{^{P-1}}+\tilde\beta Z\sqrt{\rho\left(\dfrac{P}{2}\m^{^{P-1}}\right)^2+ \alpha_{P-1}\dfrac{P!}{2}(1+\rho_P) \dfrac{P}{2} \q^{^{P-1}}}} + \tilde\beta x\right)\right].
    \\\\
    \m&=& \mathbb{E}_{_Z}\left[\tanh\left({ \tilde\beta \dfrac{P}{2}\m^{^{P-1}}+\tilde\beta Z\sqrt{\rho\left(\dfrac{P}{2}\m^{^{P-1}}\right)^2+    \alpha_{P-1}\dfrac{P!}{2}(1+\rho_P) \dfrac{P}{2} \q^{^{P-1}}}} + \tilde\beta x\right)\right].
    \end{array}
    \label{eq_self_con_x}
\end{equation}
Also, we notice that, as $\tilde \beta \to \infty$, in the previous equations $\q\to 1$, thus in order to correctly perform the limit we introduce the reparametrization 
\begin{equation}
    \q=1-\dfrac{\delta \q}{\tilde\beta}\;\;\;\;\mathrm{as}\;\;\;\;\tilde\beta\to \infty.
    \label{eq_reparam}
\end{equation}
Using \eqref{eq_reparam} in \eqref{eq_self_con_x} we obtain
\begin{equation}
\small
    \begin{array}{lll}
    \m&=& \mathbb{E}_{_Z}\left[\tanh\left({ \tilde\beta \dfrac{P}{2}\m^{^{P-1}}+\tilde\beta Z\sqrt{\rho\left(\dfrac{P}{2}\m^{^{P-1}}\right)^2+ \alpha_{b}
   \alpha_{P-1}\dfrac{P!}{2}(1+\rho_P) \dfrac{P}{2} \,\left( 1-\dfrac{\delta \q}{\tilde\beta} \right)^{^{P-1}}}} + \tilde\beta x\right)\right].
    \\\\
    1-\dfrac{\delta \q}{\tilde\beta}&=&\mathbb{E}_{_Z}\left[\tanh{}^{\2}\left({ \tilde\beta \dfrac{P}{2}\m^{^{P-1}}+\tilde\beta Z\sqrt{\rho\left(\dfrac{P}{2}\m^{^{P-1}}\right)^2+  \alpha_{P-1}\dfrac{P!}{2}(1+\rho_P) \dfrac{P}{2} \,\left( 1-\dfrac{\delta \q}{\tilde\beta} \right)^{^{P-1}}}} + \tilde\beta x\right)\right].
    \end{array}
\end{equation}
\normalsize
Taking advantage of the new parameter $x$, we can recast the last equation in $\delta \q$ as a derivative of the magnetization $\m$ :
\begin{equation}
    \dfrac{\partial\m}{\partial x}=\tilde\beta\left[1-\left(1-\dfrac{\delta\q}{\tilde\beta}\right)\right]=\delta\q.
\end{equation}
Thanks to this correspondence between the self equation for $\m$ and the one for $\delta \q$, we can focus only on the self equation for $\m$ and we can proceed with the  $\tilde\beta\to \infty$. Thus, as $\tilde\beta\to \infty$, we have 
\begin{equation}
    \begin{array}{lll}
         \m = \mathbb{E}_{_Z}\left[\mathrm{sign}\left(\dfrac{P}{2}\m^{^{P-1}}+Z\sqrt{\rho\left(\dfrac{P}{2}\m^{^{P-1}}\right)^2+  \alpha_{P-1}\dfrac{P!}{2}(1+\rho_P) \dfrac{P}{2}\;}\right)\right]
          \\\\
         \q\to 1.
     \label{eq_self_beta_finali}
    \end{array}
\end{equation}
Where we have restored $x$ to zero. Finally, we can rearrange \eqref{eq_self_beta_finali}  using the relation
\begin{equation}   \mathbb{E}_z\mathrm{sign}[b_1+z\,b_2]=\mathrm{erf}\left[\dfrac{b_1}{\sqrt{2}b_2}\right],
\end{equation}
where $b_1=\frac{P}{2}\m^{^{P-1}}$ and  $b_2= \sqrt{\rho\left(\frac{P}{2}\m^{^{P-1}}\right)^2+ 
\alpha_{P-1}\frac{P!}{2}(1+\rho_P) \frac{P}{2} \;}$, hence reaching the thesis. 
\end{proof}

\section{Plefka's expansion of the effective Gibbs potential}
\label{app:plefka}
When dealing with dense networks, MC simulations become prohibitively slow due to the computationally expensive update of their synaptic tensor (see the cost function, Eq. \ref{def:H_PHopEx}) and, clearly, the higher the interaction order $P$ the slower their convergence. To speed up simulations an alternative route where these computations can be avoided should be walked. Implementing Plefka's dynamics within a MC scheme can be a way out as the latter is an effective dynamics that allows us to keep track of the evolution of the network order-parameters at the level of their mean values.
\newline
The purpose of this section is thus to follow the path developed in \cite{Plefka1,Plefka2}, namely we switch from the free energy to the Gibbs potential via Legendre transform, then we compute the expansion of the Gibbs potential that allows us to write effective (coarse grained) dynamics for the mean of the Mattis magnetization and we use such an expression within the update rule of the MC iterations.
\newline 
We split the following analysis in two parts: first, we show the computation in classic dense networks, then we generalize the approach for unsupervised  learning. 

\subsection{Plefka's effective dynamics for Hebbian storing}
\label{app:plefka_hebb}

The system is described by its Hamiltonian
\begin{align}
\label{eq:Hintdense}
  \mathcal{H}_{N,K}^{(P)}(\bm \sigma \vert \bm \xi; \bm z)
  = - \dfrac{1}{\sqrt{ N^{P-1}}} \sum_{\mu} \sum_{i_1 \hdots i_{P/2} } \xi_{i_1}^\mu \hdots \xi_{i_{P/2}}^\mu \sigma_{i_1} \hdots \sigma_{i_{P/2}} z_\mu 
\end{align}
which is the interaction part of the integral representation of the Hamiltonian of the dense Hebbian network in \cite{Albanese2021}. Moreover, $\{z_\mu \}$ is an additional set of real Gaussian variable we have added to describe the model.

In order to use Plefka's dynamic, we introduce a control parameter $\NotAlpha$ and we defined a new Hamiltonian as 
\begin{align}
\label{Halpha}
\mathcal{H}_{N,K}^{(P)}(\boldsymbol{\sigma} \vert \bm \xi; \bm z, \bm h, \tilde{\bm h}, \NotAlpha)
= \NotAlpha \mathcal{H}_{N,K}^{(P)}(\bm \sigma \vert \bm \xi; \bm z)
- \dfrac{1}{\sqrt{N^{P-1} }}\sum_i h_i \sigma_i - \sum_{\mu} \tilde h_{\mu} z_\mu 
\end{align}
in such a way that if $\NotAlpha=0$ we have an Hamiltonian representing non-interacting units, whereas if $\NotAlpha=1$ we have an Hamiltonian representing full-interacting units; for this reason $\NotAlpha$ is referred to as interaction strength. 
Moreover, $\{ h_i \}$ and $\{\tilde{h}_i\}$ are external fields which act, respectively, on $\bm \sigma$ and $\bm z$. 

Using the expression of the modified Hamiltonian \eqref{Halpha} we write down the partition function as
\begin{align}
      &\mathcal{Z}_{N,K,\beta}^{(P)}( \bm \xi ; \bm h,  \tilde{ \bm h}, \NotAlpha )
     = \sum_{\sigma} \int \prod_\mu dz_\mu \sqrt{\frac{\b}{2\pi}} \exp \left( - \NotAlpha \dfrac{\b}{\sqrt{N^{P-1}}} \sum_{\mu} \sum_{i_1 \hdots i_{P/2} } \xi_{i_1}^\mu \hdots \xi_{i_{P/2}}^\mu \sigma_{i_1} \hdots \sigma_{i_{P/2}} z_\mu  \right.\notag \\
    &\left. - \frac{\b}{2} \sum_{\mu} z_\mu^2+\b \sum_{\mu} \tilde h_{\mu} z_\mu +\dfrac{\b}{\sqrt{N^{P-1}}} \sum_i h_i \sigma_i \right).
\end{align}
\normalsize

where $\beta'=2\beta/P!$. 

We now consider the Gibbs potential for this model, that is the Legendre transformation of the free energy constrained w.r.t. the magnetizations $m_i= \langle \si \rangle$ and $\langle z_\mu \rangle$ averaged w.r.t. Boltzmann distribution $\propto \exp^{-\beta H(\NotAlpha)}$;
for our model the Gibbs potential is given by 
\begin{align}
\label{gibbsPotential_dense}
    \mathcal G_{N,K,\beta}^{(P)}( \bm \xi ;\bm z,\bm h,  \tilde{ \bm h}, \NotAlpha )
     = -\frac{1}{\b}\ln  \mathcal Z(\NotAlpha) +  \sum_{\mu} \tilde h_{\mu} \langle z_\mu \rangle + \dfrac{1}{\sqrt{N^{P-1}}}\sum_i h_i m_i.
\end{align}

Now, we shorten the notation as $ \mathcal G _{N,K,\beta}^{(P)}( \bm \xi, ;  \bm z, \bm h,  \tilde{ \bm h}, \NotAlpha ) \equiv \mathcal G(\NotAlpha )$, and expand the
last expression around $\NotAlpha=0$ as 
\begin{align}
     \mathcal G(\NotAlpha)= \mathcal G(0) + \sum_{n=1}^\infty \frac{\NotAlpha^n G^{(n)}}{n!} \ \ \ \  \mathcal G^{(n)}= \left. \frac{\partial^n  \mathcal G(\NotAlpha)}{\partial \NotAlpha^n} \right \vert_{\NotAlpha=0} . 
     \label{eq:gibbexp}
\end{align}
For our computations we stop the expansion to the first order and we start to find all the terms we need. The non-interacting Gibbs potential $\mathcal G(0)$ reads as
\begin{align}
     \mathcal G(0)= -\frac{N}{\beta} \log 2 -\frac{1}{\b} \sum_i \log \cosh \left(\dfrac{\b}{\sqrt{N^{P-1}}} h_i\right) +\sum_{\mu} \tilde h_{\mu} \langle z_\mu \rangle  + \dfrac{1}{\sqrt{N^{P-1}}}\sum_i h_i \langle \sigma_i \rangle -\frac{1}{2} \sum_{\mu} \tilde h_{\mu}^2.
     \label{eq:G0}
\end{align}
If we extremize $\mathcal G(0)$ w.r.t. local fields, namely $h_i$ and $\tilde h_\mu$ for $\mu=1, \hdots, K,\ i=1, \hdots , N$, we find expressions of them read as 
\begin{align}
\label{eq:hi}
    h_i&= \frac{\sqrt{N^{P-1}}}{\b} {\tanh^{-1}}(m_i) \Longrightarrow h_i=\dfrac{\sqrt{N^{P-1}}}{2\b}\ln\left(\dfrac{1+m_i}{1-m_i}\right),\\
    \label{eq_hmu}
    \tilde h_\mu &= \langle z_\mu \rangle;
\end{align}
where we have use the relation $$\tanh^{-1}(x)= \ln \dfrac{1+x}{1-x}.$$ 
Putting \eqref{eq:hi} and \eqref{eq_hmu} in the non-interacting Gibbs potential \eqref{eq:G0} we get
\begin{align}
    \mathcal  G(0)= -\frac{N}{\b} \log 2 +\frac{1}{2\b} \sum_i \left[(1-m_i) \log(1-m_i) + (1+m_i) \log (1+m_i)\right]+\frac{1}{2} \sum_{\mu}  \langle z_\mu \rangle^2.
\end{align}

Now we have found $\mathcal G(0)$, all we need is the first-order contribution which is
\begin{align}
    \frac{\partial  \mathcal G(\NotAlpha)}{ \partial \NotAlpha}\Bigg|_{\NotAlpha=0} = -\dfrac{1}{\sqrt{N^{P-1}}}\sum_{i_1, \hdots, i_{P/2}} \sum_\mu \xi_{i_1}^\mu \hdots \xi_{i_{P/2}}^\mu \langle z_\mu \rangle m_{i_1} \hdots m_{i_{P/2}}.
\end{align}

Therefore the first-order expression of Gibbs potential for the full interacting system, namely for $\NotAlpha =1$ is 
\begin{align}
    \mathcal G(\NotAlpha=1)=&  -\frac{N}{\b} \log 2 +\frac{1}{2\b} \sum_i \left[(1-m_i) \log(1-m_i) + (1+m_i) \log (1+m_i)\right]+\frac{1}{2} \sum_{\mu}  \langle z_\mu \rangle^2\notag \\
    &-\dfrac{1}{\sqrt{N^{P-1}}}\sum_{i_1, \hdots, i_{P/2}} \sum_\mu \xi_{i_1}^\mu \hdots \xi_{i_{P/2}}^\mu \langle z_\mu \rangle m_{i_1} \hdots m_{i_{P/2}}.
    \label{eq:G1}
\end{align}

Extremizing \eqref{eq:G1} w.r.t. $m_i$ and $\langle z_\mu \rangle$, we find the respective self-consistency equations:

\begin{align}
\label{eq:P_a}
    \frac{\partial  \mathcal G}{\partial m_i}=0 \Rightarrow 
    m_i &= \tanh{\left[\b \dfrac{P}{2}\dfrac{1}{\sqrt{N^{P-1}}}\sum_\mu \xi_i^\mu \langle z_\mu \rangle\left(\sum_j \xi_j^\mu m_j \right)^{P/2-1}\right]},
\end{align}

\begin{align}
\label{eq:P_b}
     &\frac{\partial  \mathcal G}{\partial \langle z_\mu \rangle}=0 \Rightarrow 
     \langle z_\mu \rangle =\frac{1}{\sqrt{N^{P-1}}} \left(\sum_i \xi_i^\mu m_i\right)^{P/2}
\end{align}

These equations are then used ``in tandem'' to make the system evolve: starting from an initial configuration $(\boldsymbol \sigma^{(0)},\boldsymbol z^{(0)})$, we evaluate the related $m_i^{(0)}$ and $z_{\mu}^{(0)}$ for any $i$ and $\mu$, and we use them in \eqref{eq:P_a} to get $m_i^{(1)}$, the latter is then used in \eqref{eq:P_b} to get $z_{\mu}^{(1)}$, and we proceed this way, bouncing from \eqref{eq:P_a} to \eqref{eq:P_b}, up to thermalization.
We stress that these equations allow to implement an effective dynamics that avoid the computation of spin configurations. In fact, this coarse grained dynamics only cares about the Boltzmann average of each spin direction, whose behaviour is given by \eqref{eq:P_a}.
Even if at first glance \eqref{Halpha} seems to require three set of auxiliary variables, 
$\{z_{\mu} \}$ , $\{h_i\}$ and $\{\tilde{h}_i\}$,  the extremization of the Gibbs potential at first order fixes the external fields $\{h_i\}$ and $\{\tilde{h}\}$. The gaussian variables $\{z_{\mu} \}$ act  as latent dynamical variables that evolve according to \eqref{eq:P_b}. In such an iterative MC scheme these hidden degrees of freedom are suitably updated in order to effectively retrieve the pattern that constitutes the signal.

\subsection{Plefka's effective dynamics for unsupervised Hebbian learning}
The system is described by the Hamiltonian 
\begin{align}
         \mathcal{H}_{N,K,M, r}^{(P)}(\bm \sigma \vert \bm \eta; \bm z)
         = - \sqrt{\dfrac{1}{N^{P-1}\R^{P/2} M}}\SOMMA{\mu>1}{K}\SOMMA{a=1}{M}\left(  \SOMMA{i_{_1},\cdots, i_{_{P/2}}}{N,\cdots,N}\eta^{\mu\,,a}_{i_1}\cdots\eta^{\mu\,,a}_{i_{P/2}}\sigma_{_{i_1}}\cdots\sigma_{_{i_{_{P/2}}}}\right)\,z_{\mu,a} 
\end{align}
which is the Hamiltonian corresponding of the interacting term in \eqref{eq:integral}. Moreover,  $\{z_{\mu,a} \}$ is an additional set of real variable computed by Gaussian distribution we have added to describe the model.
Mirroring computations in Subsection \ref{app:plefka_hebb}, in order to use Plefka's dynamic, we introduce a control parameter $\NotAlpha$ and we defined a new Hamiltonian as

\begin{align}
     \mathcal{H}_{N,K,M, r}^{(P)}(\bm \sigma \vert \bm \eta; \bm z,\bm  h,  \tilde{ \bm h}, \NotAlpha) = \NotAlpha  \, 
     \mathcal{H}_{N,K,M}^{(P)}(\bm \sigma \vert \bm \eta;\bm z)
     - \sqrt{\dfrac{1}{N^{P-1}\R^{P/2} M}} \sum_i h_i \sigma_i -\sum_{\mu,a} \tilde h_{\mu,a} z_{\mu,a} 
     \label{Halpha_sup}
\end{align}
where $\NotAlpha$ describes the interaction strength. 
We stress that if $\NotAlpha=0$ we have the Hamiltonian of non-interacting terms. Moreover, $\{h_i\}$ and $\{\tilde{h}_{\mu, a}\}$ are external fields who act, respectively, on $\bm \sigma$ and $\bm z$. 

The expression of the modified Hamiltonian \eqref{Halpha_sup} can be used to write down the partition function as
\begin{align}
      &\mathcal{Z}_{N,K,M,r, \beta}^{(P)}( \bm \eta; \bm h,  \tilde{ \bm h}, \NotAlpha )
      =\sum_{\sigma} \int \prod_{\mu,a} dz_{\mu,a} \sqrt{\frac{\bbt}{2\pi}} \exp \left( - \frac{\bbt }{2} \sum_{\mu,a} z_{\mu,a}^2 +\bbt \sum_{\mu,a} \tilde h_{\mu,a} z_{\mu,a}  \right.\notag \\
    &\left.+\dfrac{\bbt}{\sqrt{N^{P-1}r^{2P} M}}\sum_i h_i \sigma_i  - \NotAlpha \dfrac{\bbt}{\sqrt{N^{P-1}r^{2P} M}} \sum_{\mu,a} \sum_{i_1, \hdots , i_{P/2}} \eta_{i_1}^{\mu,a} \hdots \eta_{i_{P/2}}^{\mu,a} \sigma_{i_1} \hdots \sigma_{i_{P/2}} z_{\mu,a} \right)
\end{align}
where we define $\tilde{\beta}$ as in \eqref{eq:tildebeta}. 

The Gibbs potential is defined as the Legendre transformation of the free energy constrained w.r.t. the magnetization $m_i= \langle \si \rangle$ and $\langle z_{\mu,a} \rangle$ averaged w.r.t. the Boltzmann distribution $P(\bm \sigma) \sim \exp^{-\beta H(\NotAlpha)}$  \eqref{gibbsPotential_dense}. We have 
\begin{align}
\label{gibbsPotential}
     &\mathcal G _{N,K,M,r, \beta}^{(P)}( \bm \eta;\bm z,\bm h,   \tilde{ \bm h}, \NotAlpha )
      -\frac{1}{\bbt}\ln  
     \mathcal{Z}_{N,K,M,r, \beta}^{(P)}( \bm \eta; \bm h,   \tilde{ \bm h}, \NotAlpha )
     +  \sum_{\mu} \tilde h_{\mu} \langle z_\mu \rangle + \dfrac{1}{\sqrt{N^{P-1}r^{2P} M}} \sum_i h_i m_i
\end{align}
and we write Plefka's expansion as in \eqref{eq:gibbexp}. Also in this case we stop the expansion at the first order. Thus,
shortening the notation as $ \mathcal G _{N,K,M,r, \beta}^{(P)}( \bm \eta;  \bm z, \bm h, \tilde{\bm h}, \NotAlpha ) \equiv G(\NotAlpha)$, we compute the non-interacting Gibbs potential $\mathcal G(0)$ and the first order contribution $\left.\dfrac{\partial \mathcal{G}(\NotAlpha)}{\partial \NotAlpha}\right\vert_{\NotAlpha=0}$. Let us start from $\mathcal G(0)$. 

\begin{align}
     \mathcal G(0)=& -\frac{N}{\bbt} \log 2 -\frac{1}{\bbt} \sum_i \log \cosh \left(\dfrac{\bbt}{\sqrt{N^{P-1}r^{2P} M}} h_i\right) +\sum_{\mu,a} \tilde h_{\mu,a} \langle z_{\mu,a}\rangle  \notag \\
     &+ \dfrac{1}{\sqrt{N^{P-1}r^{2P} M}}\sum_i h_i \langle \sigma_i \rangle -\frac{1}{2} \sum_{\mu,a} \tilde h_{\mu,a}^2
\end{align}

If we extremize $ \mathcal G(0)$ w.r.t. the local fields, namely $h_i$ and $\tilde{h}_{\mu,a}$ for $i=1, \hdots, N$, $a=1, \hdots , M$, we can find their expressions, that read as 
\begin{align}
    h_i&=\dfrac{\sqrt{N^{P-1}r^{2P} M}}{2\bbt}\log\left(\dfrac{1+m_i}{1-m_i}\right),\\
    \tilde h_{\mu,a} &= \langle z_{\mu,a} \rangle;
\end{align}
While the first-order derivative of Gibbs potential w.r.t. $\NotAlpha$ is 
\begin{align}
    \left.\frac{\partial  \mathcal G(\NotAlpha)}{ \partial \NotAlpha}\right\vert_{\NotAlpha=0} = -\dfrac{1}{\sqrt{N^{P-1}r^{2P} M}}\sum_{i_1, \hdots, i_{P/2}} \sum_{\mu,a} \eta_{i_1}^{\mu,a} \hdots \eta_{i_{P/2}}^{\mu,a} \langle z_{\mu,a} \rangle m_{i_1} \hdots m_{i_{P/2}}.
\end{align}
Therefore, $ \mathcal G(\NotAlpha)$ is rewritten using Plefka's expansion as 
\begin{align}
     \mathcal G(\NotAlpha)=&-\frac{N}{\bbt} \log 2 +\frac{1}{2\bbt} \sum_i \left[(1-m_i) \log(1-m_i) + (1+m_i) \log (1+m_i)\right] \notag \\
    &+\frac{1}{2} \sum_{\mu,a}  \langle z_{\mu,a} \rangle^2 -\dfrac{\NotAlpha}{\sqrt{N^{P-1}r^{2P} M}}   \sum_{\mu,a} \left(\sum_{i}\eta_{i}^{\mu,a} m_{i} \right)^{P/2} \langle z_{\mu,a} \rangle.
\end{align}

To conclude, we can compute the self-consistence equations w.r.t. $m_i$ and $\langle z_{\mu, a}\rangle$ extremizing the first order expression of $\mathcal{G}(\varphi =1)$ read as 
\begin{align}
    &m_i = \tanh{\left[\bbt \dfrac{P}{2}\dfrac{1}{\sqrt{N^{P-1}r^{2P}M}}\sum_{\mu,a} \eta_i^{\mu,a} \langle z_{\mu,a} \rangle\left(\sum_j \eta_j^{\mu,a} m_j \right)^{P/2-1}\right]}, \\
     &\langle z_{\mu,a} \rangle =\frac{1}{\sqrt{N^{P-1}r^{2P}M}} \left(\sum_i \eta_i^{\mu,a} m_i\right)^{P/2}
\end{align}
 
These equations are then used ``in tandem'' to make the system evolve, as explained in the previous subsection.
This iterative MC updating scheme leaves the network free to arrange the hidden degrees of freedom $\{ z_{\mu,a} \}$ in such a way that the mean values of the neurons $m_i$ converge to the correspondent element of the archetype vector, provided that the network is posed in the retrieval region of the phase diagram.

\section{Evaluation of the momenta of the effective post-synaptic potential}
\label{app:momenta}

The purpose of this section is to evaluate the first and second momenta of the expression 
$\xi_i^1 h_i^{(1)}(\boldsymbol \xi^1)$, that are referred to as, respectively, $\mu_1$ and $\mu_2$ and are used in 
Sec.~\ref{sec:MC}; specifically,
\begin{eqnarray}
\small
\mu_1\coloneqq \mathbb{E}_{\xi}\mathbb{E}_{(\eta|\xi)}\left[\xi_i^{1}h_i^{(1)}(\boldsymbol{\xi}^{1})\right]=\dfrac{1}{\R^{P/2}M N^{P-1}}\SOMMA{\mu,a=1}{K,M}\SOMMA{(i_2,\cdots,i_P)}{}\mathbb{E}_{\xi}\mathbb{E}_{(\eta|\xi)}\left[(\xi^1_{i_1}\cdots\xi^{1}_{i_P})\eta^{\mu, a}_{i_1}...\eta^{\mu,a}_{i_P}\right],
\end{eqnarray}
\begin{align}
\begin{array}{llll}
     \mu_2\coloneqq \mathbb{E}_{\xi}\mathbb{E}_{(\eta|\xi)}\left[\lbrace h_i^{(1)}(\boldsymbol{\xi}^{1})\rbrace ^2\right]&=\dfrac{1}{ N^{2P-2}\R^{P}M^2}\SOMMA{\mu,\nu=1}{K}&\SOMMA{a,b=1}{M,M}\SOMMA{(i_2,\cdots,i_P)}{}\SOMMA{(j_2,\cdots,j_P)}{}\mathbb{E}_{\xi}\mathbb{E}_{(\eta|\xi)}\left[\left(\xi^1_{i_2}\xi^1_{j_2}...\xi^1_{i_P}\xi^1_{j_P}\right)\right.
\\\\
&&\left.\eta^{\mu, a}_{i_1}\eta^{\nu, b}_{i_1}\left(\eta^{\mu, a}_{i_2}\eta^{\nu, b}_{j_2}...\eta^{\mu,a}_{i_P}\eta^{\nu,b}_{j_P}\right)\right].
\end{array}
\end{align}

As for $\mu_1$, using $\mathbb{E}_{(\eta|\xi)}[\eta_i^{\mu,a}]=r\xi_i^\mu$,
\begin{eqnarray}
\small
\mu_1=\dfrac{1}{\R^{P/2}M N^{P-1}}\SOMMA{\mu=1}{K}\SOMMA{(i_2,\cdots,i_P)}{}\mathbb{E}_{\xi}\left[Mr^P\left(\xi^1_{i_1}\cdots\xi^{1}_{i_P}\right)\left(\xi^{\mu}_{i_1}...\xi^{\mu}_{i_P}\right)\right],
\end{eqnarray}
since $\mathbb{E}_{\xi}[\xi_i^{\mu}]=0$ the only non-zero terms are those with $\mu=1$ and the expression simplifies into
\begin{eqnarray}
\mu_1
= \dfrac{ r^P}{\R^{P/2} N^{P-1}}\SOMMA{(i_2,\cdots,i_P)}{}\mathbb{E}_{\xi}\left[\left(\xi^1_{i_1}\cdots\xi^{1}_{i_P}\right)^2\right]=\dfrac{1}{(1+\rho)^{P/2}} .
\end{eqnarray}

As for $\mu_2$, 
since $\mathbb{E}_{\xi}\mathbb{E}_{(\eta|\xi)}[\eta^{\mu,a}_{i_1}\eta^{\nu,b}_{i_1}]=r^2\delta^{\mu\nu}$, the only non-zero terms are those where $\mu=\nu$, thus
\begin{align}
\begin{array}{llll}
     \mu_2&= &\dfrac{1}{N^{2P-2}\R^{P}M^2}\SOMMA{\mu=1}{K}\SOMMA{a,b=1}{M,M}\SOMMA{(i_2,\cdots,i_P)}{}\SOMMA{(j_2,\cdots,j_P)}{}\mathbb{E}_{\xi}\mathbb{E}_{(\eta|\xi)}\left[\left(\xi^1_{i_2}\xi^1_{j_2}...\xi^1_{i_P}\xi^1_{j_P}\right)\right.
\\\\
&&\left.\eta^{\mu, a}_{i_1}\eta^{\mu, b}_{i_1}\left(\eta^{\mu, a}_{i_2}\eta^{\mu, b}_{j_2}...\eta^{\mu,a}_{i_P}\eta^{\mu,b}_{j_P}\right)\right]
\\\\
&=&A_{\mu=1}+B_{\mu>1},
\end{array}
\end{align}
\normalsize
where, in the last line, we highlighted the contributions stemming from terms with, respectively, $\mu=1$ ($A_{\mu=1}$) and $\mu>1$ ($B_{\mu >1}$). These are evaluated hereafter:
\begin{align}
\small
\label{eq:amu}
\begin{array}{llll}
    A_{\mu=1}&=&\dfrac{1}{ N^{2P-2}\R^{P}M^2}\SOMMA{\mu=1}{K}\SOMMA{(i_2,\cdots,i_P)}{}\SOMMA{(j_2,\cdots,j_P)}{}\left\{\SOMMA{a=1}{M}\mathbb{E}_{\xi}\mathbb{E}_{(\eta|\xi)}\left[\left(\xi^1_{i_2}\xi^1_{j_2}...\xi^1_{i_P}\xi^1_{j_P}\right)\right.\right.
\\\\
&&\left.\left(\eta^{\mu, a}_{i_2}\eta^{1, a}_{j_2}...\eta^{1,a}_{i_P}\eta^{1,a}_{j_P}\right)\right]+\left.\SOMMA{a\neq b}{M}\mathbb{E}_{\xi}\mathbb{E}_{(\eta|\xi)}\left[\left(\xi^1_{i_2}\xi^1_{j_2}...\xi^1_{i_P}\xi^1_{j_P}\right)\eta^{1, a}_{i_1}\eta^{1, b}_{i_1}\left(\eta^{\mu, a}_{i_2}\eta^{1, b}_{j_2}...\eta^{1,a}_{i_P}\eta^{1,b}_{j_P}\right)\right]\right\}
    \\\\
    &=&\dfrac{1}{ N^{2P-2}\R^{P}M^2}\SOMMA{\mu=1}{K}\SOMMA{(i_2,\cdots,i_P)}{}\SOMMA{(j_2,\cdots,j_P)}{}\mathbb{E}_{\xi}\left[\SOMMA{a=1}{M}r^{2P-2}\left(\xi^1_{i_2}\xi^1_{j_2}...\xi^1_{i_P}\xi^1_{j_P}\right)^2+\SOMMA{a\neq b}{M}r^{2P}\left(\xi_{i_1}^1\xi^1_{i_2}\xi^1_{j_2}...\xi^1_{i_P}\xi^1_{j_P}\right)^2\right]
    \\\\
    &=&\dfrac{r^{2P}}{\R^P M}\left[ r^{-2}+(M-1)\right]=\dfrac{r^{2P}}{\R^P}\left[1+ \dfrac{1-r^2}{r^2M}\right]=\left(\dfrac{1}{(1+\rho)^{P/2}}\right)^2\left(1+\rho\right)\,;
\end{array}
\end{align}
for $B_{\mu>1}$, splitting the case $a=b$ and $a\neq b$, we get
\begin{align}
\begin{array}{llll}
     B_{\mu >1}&=&\dfrac{1}{ N^{2P-2}\R^{P}M^2}\SOMMA{\mu>1}{K}\SOMMA{(i_2,\cdots,i_P)}{}\SOMMA{(j_2,\cdots,j_P)}{}\left\{\mathbb{E}_{\xi}\mathbb{E}_{(\eta|\xi)}\SOMMA{a=1}{M}\left[\left(\xi^1_{i_2}\xi^1_{j_2}...\xi^1_{i_P}\xi^1_{j_P}\right)\right.\right.
\\\\
&&\left.\left(\eta^{\mu, a}_{i_2}\eta^{\mu, a}_{j_2}...\eta^{\mu,a}_{i_P}\eta^{\mu,a}_{j_P}\right)\right]+\mathbb{E}_{\xi}\mathbb{E}_{(\eta|\xi)}\SOMMA{a\neq b}{M}\left(\xi^1_{i_2}\xi^1_{j_2}...\xi^1_{i_P}\xi^1_{j_P}\right)\left.\eta^{\mu, a}_{i_1}\eta^{\mu, b}_{i_1}\left(\eta^{\mu, a}_{i_2}\eta^{\mu, b}_{j_2}...\eta^{\mu,a}_{i_P}\eta^{\mu,b}_{j_P}\right)\right]\Bigg\}
\\\\
&=&\dfrac{1}{ N^{2P-2}\R^{P}M^2}\SOMMA{\mu>1}{K}\SOMMA{(i_2,\cdots,i_P)}{}\SOMMA{(j_2,\cdots,j_P)}{}\mathbb{E}_{\xi}\left\{r^{2P-2}\SOMMA{a=1}{M}\left(\xi^1_{i_2}\xi^1_{j_2}...\xi^1_{i_P}\xi^1_{j_P}\right)\left(\xi^{\mu}_{i_2}\xi^{\mu}_{j_2}...\xi^{\mu}_{i_P}\xi^{\mu}_{j_P}\right)\right.
\\\\
&&+r^{2P}\SOMMA{a\neq b}{M}\left(\xi^1_{i_2}\xi^1_{j_2}...\xi^1_{i_P}\xi^1_{j_P}\right)\left(\xi^{\mu}_{i_2}\xi^{\mu}_{j_2}...\xi^{\mu}_{i_P}\xi^{\mu}_{j_P}\right)\Bigg\}
\end{array}
\end{align}
as far as $\mathbb{E}_\xi(\xi^\mu_i\xi^\mu_j)=\delta_{ij}$, the only non-zero terms are the ones in which the sum over $i$ and $j$ will be equal in pairs:
\begin{align}
\begin{array}{llll}
     B_{\mu >1}&=&\dfrac{(P-1)!}{N^{2P-2}\R^{P}M^2}\SOMMA{\mu>1}{K}\SOMMA{(i_2,\cdots,i_P)}{}\mathbb{E}_{\xi}\left\{\left[r^{2P-2}M+r^{2P}M(M-1)\right]\left(\xi^1_{i_2}...\xi^1_{i_P}\right)^2\left(\xi^{\mu}_{i_2}...\xi^{\mu}_{i_P}\right)^2\right\}
     \\\\
     &=&\dfrac{r^{2P}}{N^{P-1} \R^P}(P-1)!K\left(1+\dfrac{1-r^{2P}}{Mr^{2P}}\right)=\dfrac{r^{2P}}{N^{P-1}\R^P }(P-1)!K\left(1+\rho_{_P}\right),
\end{array}
\end{align}
and, if we 
set $K=\alpha_{P-1}N^{P-1}$ we have  
\begin{align}
\begin{array}{llll}
     B_{\mu >1}=\left(\dfrac{1}{(1+\rho)^{P/2}}\right)^2(P-1)!\alpha_{P-1}\left(1+\rho_{_P}\right).
     \label{eq:bmu}
\end{array}
\end{align}

Putting together \eqref{eq:amu} and \eqref{eq:bmu} we reach the explicit expression of $\mu_2$.


\section{List of symbols (in alphabetical order)}
\begin{itemize}
    \item $\mathcal{A}$ is the statistical quenched pressure (i.e. $\mathcal{A}=-\beta  \mathcal{F}$)
    \item $\alpha_{b}$ is the storage of the network defined as $\alpha_{b}=\lim_{N\rightarrow +\infty} K/N^b$
    \item $\mathcal{B}(\bm \sigma; t)$ is the Boltzmann factor, defined as $\mathcal{B}(\bm \sigma; t)=\exp[\beta {H}(\bm \sigma; t)]$
    \item $\beta \in \mathbb{R}^+$   is  the (fast) noise in the network (such that for $\beta \to 0$ the behavior of the network is a pure random walk while for $\beta \to +\infty$ it is a steepest descent toward the minima)
    \item $\mathbb{E}$ denotes the average over all the (quenched) coupling variables
    \item  $\bm \eta \in \{ -1, +1 \}^{K\times M}$ are the noisy examples, namely  noisy versions of the archetypes $\bm \xi^{\mu}$
    \item $\mathcal{F}$ is the free energy (i.e. $\mathcal{F}= -\beta^{-1}\mathcal{A}$)
    \item $\gamma$ is the $P$ independent part of $\alpha_{b}$, namely $\alpha_{b}= \gamma \frac{2}{P!}$
    \item $\mathcal{H}$ is the cost function (or Hamiltonian) defining the model
    \item $K \in \mathbb{N}$ is the amount of archetypes $\bm \xi$ to learn and retrieve
    \item $N \in \mathbb{N}$ is the amount of neurons in the network, i.e. the network size
    \item $M \in \mathbb{N}$ is the amount of examples per archetype, i.e. the training set
    \item $m_\mu$ is the Mattis magnetization of the archetype $\xi^\mu$ defined as $\frac{1}{N} \sum_i \xi_i^{\mu} \sigma_i$
    \item $\bar{m}$ is the asymptotic value of the Mattis magnetization 
    of the signal, $m$, 
    in the thermodynamic limit, i.e. $\lim\limits_{N \to \infty}P(m_{\mu})=\delta(m
    - \bar{m})$
    \item $n_{a, \mu}$ is the magnetization of the example $\eta^{\mu,a}$ defined as $\frac{1}{N} \sum_i \eta_i^{\mu,a} \sigma_i$
    \item $\bar{n}$ is the asymptotic value of the magnetization 
    of each example related to the signal, $n_{1, a}$,  
    in the thermodynamic limit, i.e.  $\lim_{N \to \infty}P(n_{1,a})=\delta(n_{1,a} - \bar{n})$
    \item $\omega_t(O(\bm \sigma))$ is the generalized Boltzmann measure, namely $\omega_t(O(\bm \sigma))= \frac{1}{\mathcal{Z}_N} \sum_{\bm \sigma} O(\bm \sigma) \mathcal{B}(\bm \sigma; t)$
    \item $w$ is the noise added to synaptic tensor $\bm \eta$ defined as $w = \tau N^{\delta}$ where $\delta=\frac{(P-1)-b}{2}$ 
    \item $P$ is the degree of interaction among neurons in the network (e.g. $P=2$ is the standard pairwise scenario)
    \item $q_{lm}$ is the overlap among two replicas  defined as $\frac{1}{N} \sum_i \sigma_i^{(l)} \sigma_i^{(m)}$ 
    \item $\bar{q}$ is the asymptotic value of $q_{lm}$ in the thermodynamic limit and under the replica symmetric assumption, i.e. $\lim_{N \to \infty}P(q_{lm})=\delta(q_{lm}-\bar{q})$   
    \item $\mathcal{R}$ is defined as $\mathcal{R}=r^2 + \frac{1-r^2}{M}$
    \item $r$ assesses the training set quality such that for $r\rightarrow 1$ the example matches perfectly the archetype whereas for $r=0$ solely noise remains.
    \item $\rho$ quantifies the entropy of the training set, namely $\rho=\frac{1-r^2}{r^2M}$
    \item $\rho_{_P}$ is a generalizaton of $\rho$ defined as $\rho_{_P}=\frac{1-r^{2P}}{r^{2P}M}$
    \item $t \in (0,1)$ is the parameter for Guerra's interpolation: when $t=1$ we recover the original model, whereas for $t=0$ we compute the one-body terms
    \item $\mathcal{Z}$ is the partition function
    \item $\langle O(\bm \sigma) \rangle$ is the generalize average defined as $\langle O(\bm \sigma) \rangle = \mathbb{E} \omega_t(O(\bm \sigma))$
\end{itemize}

\section*{Acknowledgments}

E.A. acknowledges ﬁnancial support from Sapienza University of Rome (RM120172B8066CB0) and from PNRR MUR project n. PE0000013-FAIR.
\newline
A.B. acknowledges ﬁnancial support from Ministero degli Affari Esteri e della Cooperazione Internazionale Italy-Israel (F85F21006230001) and PRIN grant {\em Statistical Mechanics of Learning Machines: from algorithmic and information-theoretical limits to new biologically inspired
paradigms} n. 20229T9EAT.
\newline
L.A. acknowledges financial support from INdAM –GNFM Project (CUP E53C22001930001) and from PRIN grant {\em Stochastic Methods for Complex Systems} n. 2017JFFHS
\newline
D.L. acknowledges INdAM and C.N.R. (National Research Council),  and A.A. acknowledges  UniSalento, both for financial  support via PhD-AI. 
\newline
E.A., L.A., F.A., A.A., A.B acknowledge the stimulating research environment provided by the Alan Turing Institute's Theory and Methods Challenge Fortnights event ``Physics-informed Machine Learning". 

\bibliographystyle{abbrv}
\bibliography{Unsupervised_V5_AB_AllBlack.bib}

\end{document}